\RequirePackage{amsthm}
\documentclass[sn-mathphys-num]{sn-jnl}

\usepackage{amsmath}
\usepackage{mathrsfs}
\usepackage[nameinlink]{cleveref}
    \crefname{ex}{Example}{Examples}
    \crefname{thm}{Theorem}{Theorems} 
    \crefname{lem}{Lemma}{Lemmas}
    \crefname{prop}{Proposition}{Propositions}
    \crefname{cor}{Corollary}{Corollaries} 
    \crefname{conj}{Conjecture}{Conjectures} 
    \crefname{defn}{Definition}{Definitions}
    \crefname{rmk}{Remark}{Remarks} 

\theoremstyle{thmstyleone}

\newtheorem{theorem}{Theorem}[section]
\newtheorem{lemma}[theorem]{Lemma}
\newtheorem{proposition}[theorem]{Proposition}
\newtheorem{corollary}[theorem]{Corollary}
\theoremstyle{definition} 
\newtheorem{definition}[theorem]{Definition}
\newtheorem{example}[theorem]{Example}
\newtheorem{remark}[theorem]{Remark}	

\usepackage{anyfontsize}

\usepackage{tikz, graphicx}
\usetikzlibrary{automata,positioning}
    
	\usetikzlibrary{positioning}
	\usetikzlibrary{arrows}
        \tikzset{%
        fwdrxn/.style={very thick, arrows={-Stealth[length=5pt,width=5pt]}},
        revrxn/.style={very thick, arrows={-Stealth[length=5pt,width=5pt,left]}},
        newt/.style={turq, opacity=0.15}
        }
        \tikzset{near start abs-right/.style={xshift=1cm}}
        \tikzset{near start abs-left/.style={xshift=-3.5cm}}
        \tikzset{near start abs-up/.style={yshift=1.5cm}}
        \tikzset{near start abs-down/.style={yshift=-1cm}}
	
\usepackage{color, xcolor}

	\definecolor{orange}{RGB}{250, 140, 0}
		
	\definecolor{turq}{RGB}{0, 160, 160}
		
	\definecolor{violet}{RGB}{164, 98, 234}
		
    \definecolor{viridisyellow}{RGB}{253,231,36}
    \definecolor{viridisyellowpale}{RGB}{239,223,81}
    \definecolor{viridisgreen}{RGB}{121,209,81}
        \definecolor{hlgreen}{RGB}{16,115,16}
    \definecolor{viridisturq}{RGB}{34,167,132}
    \definecolor{viridisblue}{RGB}{64,67,135}
    \definecolor{viridisviolet}{RGB}{68,1,84}
	
	\definecolor{ratecnst}{RGB}{172,172,172}
	
\usepackage{multirow}
\usepackage{array}

\usepackage{gauss}

\usepackage{amsfonts, amssymb, mathtools}
	
\newcommand{\be}{\begin{equation}}
\newcommand{\ee}{\end{equation}}

\newcommand{\rr}{\ensuremath{\mathbb{R}}}   
 
\newcommand{\zz}{\ensuremath{\mathbb{Z}}}

\newcommand{\p}{\partial}

\renewcommand{\epsilon}{\varepsilon}	

\newcommand{\vv}[1]{{\boldsymbol{#1}}}  

\newcommand{\Gk}{\ensuremath{(G, \bk)}}
\newcommand{\GkI}{\ensuremath{(G, \bk, \II)}}

\newcommand{\ratecnst}[1]{{\footnotesize{\color{blue}{#1}}}}

\usepackage{enumitem}
\usepackage{chemfig} 

\usepackage{float}
\usepackage{pgfplots}

\newcommand\II{{\mathcal I}}
\newcommand\RR{\mathbb{R}}
\newcommand\GG{\mathcal{G}}
\newcommand\ZZ{\mathbb{Z}}

\newcommand\by{\boldsymbol{y}}

\renewcommand\bf{\boldsymbol{f}}

\newcommand\bk{\boldsymbol{k}}

\newcommand\bu{\boldsymbol{u}}

\newcommand\bx{\boldsymbol{x}}
\newcommand\bv{\boldsymbol{v}}

\newcommand{\mS}{\mathcal{S}}

\newcommand{\hbk}{\hat{\bk}}

\newcommand{\defi}{\textbf}
\DeclareMathOperator{\spn}{span}
\newcommand{\etal}{\textrm{et al.}}

\begin{document}

\title{
Infinitesimal Homeostasis in Mass-Action Systems
}



\author[1]{\fnm{Jiaxin} \sur{Jin}}
\email{jiaxinjjx@gmail.com}
\equalcont{These authors contributed equally to this work.}

\author*[1,2]{\fnm{Grzegorz A.} \sur{Rempala}}
\email{rempala.3@osu.edu}
\equalcont{These authors contributed equally to this work.}

\affil*[1]{\orgdiv{Department of Mathematics}, \orgname{The Ohio State University}, \orgaddress{\city{Columbus},  \state{OH}, \postcode{43210}, \country{USA}}}

\affil[2]{\orgdiv{Department of Biostatistic}, \orgname{The Ohio State University}, \orgaddress{\city{Columbus},  \state{OH}, \postcode{43210}, \country{USA}}}

\abstract{
Homeostasis occurs in a biological system when a chosen output variable remains approximately constant despite changes in an input variable.  In this work we specifically focus on biological systems which may be represented as chemical reaction networks and consider their infinitesimal homeostasis, where the derivative of the input-output function is zero. The specific challenge of chemical reaction networks is that they often obey various conservation laws complicating the standard input-output analysis.    We derive several results that allow to verify the existence of infinitesimal homeostasis points both in the absence of conservation and under conservation laws where conserved quantities serve as input parameters.  In particular,   we introduce the notion of infinitesimal concentration robustness, where the output variable remains nearly constant despite fluctuations in the conserved quantities.  We provide several examples of chemical networks which illustrate our results both in deterministic and stochastic settings. 
}






\maketitle

\newpage

\tableofcontents


\section{Introduction}
\label{sec:intro}

The original concept of homeostasis refers to a regulatory mechanism that keeps some variable close to a fixed value despite varying external factors \cite{B98}. One classic example is the regulation of body temperature in mammals, regardless of environmental temperature changes. In \cite{C26}, this idea was further developed, and the term `homeostasis' was coined. 
In modern science, the concept of homeostasis holds significant importance and is extensively studied in various fields, including molecular and population biology, biochemistry, and control engineering.

Since various biological systems are often described by differential equations, one mathematical approach to describing homeostasis is to analyze a family of stable equilibria within a parameterized dynamical system of ordinary differential equations (ODEs). Homeostasis is observed when the stable equilibrium exhibits a relatively small change in response to a significantly larger change in an external parameter.

In \cite{GS17}, Golubitsky and Stewart used singularity theory to analyze homeostasis.
They considered homeostasis in an input-output function associated with an ODE  system with an input parameter $\II$ and an output variable $x_o$. 
Assuming that the input-output function $x_o (\II)$ is well-defined in some neighborhood of a specific value $\II_0$ they introduced the concept of infinitesimal homeostasis, defined by the relation $\frac{d}{d \II} x_o (\II_0) = 0$. As the function value changes slowly near a critical point, this condition aligns with the intuitive notion of homeostasis.

In \cite{WHAG21}, Wang \etal\ studied infinitesimal homeostasis in input-output networks $\GG$, where the input node $\iota$ is affected by the input parameter $\II$, and the output node is $o$.
They considered an admissible family of ODEs associated with the input-output network $\GG$. If we suppose that it admits a linearly stable family of equilibria, then the input-output function $x_o(\II)$ is well-defined for $\GG$.  
They further demonstrated that the derivative of the input-output function, $\frac{d}{d \II} x_o (\II_0)$, depends on the determinant of the homeostasis matrix $H$ (see \eqref{def:J_and_H} below for the definition). Therefore, the infinitesimal homeostasis can be determined when $\det (H) = 0$.
In a recent paper \cite{HP23}, Duncan \etal\ considered the homeostasis pattern in an input-output network $\GG$. A homeostasis pattern is defined in \cite{HP23} as a set of nodes in $\GG$, including the output node $o$, with all nodes exhibiting infinitesimal homeostasis at a specific value $\II_0$.  

In this article we will be interested in special types of ODE systems that represent dynamics of chemical reaction networks. Founded in the 1960s, chemical reaction network theory is an area of applied mathematics used to model the behavior of chemical and biochemical systems \cite{feinberg1979lectures, feinberg2019foundations}. Due to its applications in modern biology,  it has attracted the interest of both applied and pure mathematics communities.

In one of the first papers in the subject, Craciun and Deshpande analyzed homeostasis from the perspective of reaction networks \cite{CD22}. Given a reaction network $G$, they constructed a modified network $G'$ to determine whether $G$ has the capacity for homeostasis and injectivity.
Recently, in \cite{PS24} Yu and Sontag studied infinitesimal homeostasis, also known as quasi-adaptation, from a control theory perspective. They provided necessary and sufficient conditions when minors of a symbolic matrix had mixed signs, thereby establishing necessary conditions for quasi-adaptation.

Here we will study a broader setting of infinitesimal homeostasis of reaction networks. Although some related work has already been done in \cite{CD22}, the authors there considered only a relatively simple case where the stoichiometric subspace (see Definition~\ref{def:mass_action}) of the reaction network encompassed the entire Euclidean space of the appropriate dimension. In our analysis, we also include a more natural scenario where a reaction network admits conservation laws (see Definition~\ref{def:conservation_law}).

Consider for example the reaction network consisting of a reversible pair as follows:
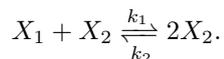
\[
X_1 + X_2 \xrightleftharpoons[k_2]{k_1} 2 X_2.
\]
Under the mass-action kinetics, the concentrations of species $x_1 (t), x_2 (t)$ satisfy
\[
\frac{d}{d t} (x_1 + x_2) = 0.
\]
Therefore, there exists a constant $C$ defined by the initial condition, such that
\[
x_1 (t) + x_2 (t) \equiv C.
\]
First, we assume the constant $C$ is fixed, the input parameter is a reaction rate constant ($k_1$ or $k_2$), and the output variable is the corresponding steady state on the species $X_2$.
In Theorem~\ref{thm:inf_homeostasis_conservtion} below, we provide a general way of verifying when infinitesimal homeostasis can happen in such case.
Next,  we extend our input-output framework and replace the input parameter with the conservation constant $C$ while keeping the same output variable.
In this setting, we refer to such infinitesimal homeostasis as {\em infinitesimal concentration robustness} (see Definition \ref{def:inf_robust_concentration}).
 In Theorem \ref{thm:inf_robust_concentration} 
 below, we provide a method to determine when infinitesimal concentration robustness occurs. This method is based on the notions of the reduced Jacobian matrix and the reduced homeostasis matrix. Both matrices incorporate information from the original Jacobian matrix and the conservation laws. Using combinatorial matrix theory, we provide a sufficient condition for when a mass-action system exhibits infinitesimal concentration robustness, while a necessary condition for infinitesimal concentration robustness follows from Theorem \ref{thm:inf_robust_concentration}.

Finally, we consider infinitesimal homeostasis in the stochastic reaction networks, where the concentrations of species are random variables.
Suppose a reaction network admits a stationary distribution, then we let the expectation of such a stationary distribution on a chosen random variable as the output.
Inspired by the deterministic case, we provide a way to find the infinitesimal concentration robustness when the reaction network is first-order (See Definition \ref{def:first_order}).

\medskip

\subsection*{Outline of the Paper.}

Section \ref{sec:background} introduces the terminology of mass-action systems, input-output networks, and infinitesimal homeostasis.
In Section \ref{ss:homeo_mas=Rn} we consider homeostasis in mass-action systems without the conservation laws. 
In Section \ref{ss:homeo_cb} we discuss when infinitesimal homeostasis can happen in a complex-balanced system. 
Section \ref{ss:reduce_jacobian} introduces the notion of a reduced Jacobian matrix.
In Section \ref{ss:deterministic_system} we consider reaction networks under the conservation laws, and state two main theorems of this paper, Theorems \ref{thm:inf_homeostasis_conservtion} and \ref{thm:inf_robust_concentration}, which provide a way to compute infinitesimal homeostasis and infinitesimal concentration robustness, respectively.
In Section \ref{ss:stochastic_system} we discuss infinitesimal homeostasis in the stochastic reaction networks.
Section \ref{sec:discussion} contains a brief discussion of some possible generalizations on different types of infinitesimal homeostasis and general stochastic reaction networks.

\medskip

\subsection*{Notation.}
Throughout this work, we let $\RR_{\geq 0}^n$ and $\RR_{> 0}^n$ denote $n$-dimensional Euclidean non-negative and positive orthants, respectively.
Similarly, $\ZZ_{\geq 0}^n$ represents the set of vectors with non-negative integer components. 
We let $\mathbf{e}_i = (0, \ldots, 0, 1, 0, \ldots, 0)$ 
represent the standard basis vector where the $1$ appears in the $i$-th position.
Vectors are typically denoted by $\bx$ or $\by$.


\section{Background}
\label{sec:background}

In this section, we will begin by introducing key terminology and results in the reaction network theory. Then we will review concepts and findings related to homeostasis in the input-output network.

\subsection{Euclidean Embedded Graphs and Mass-Action Systems}
\label{ss:eeg_mas}

In this subsection, we introduce a directed graph in $\mathbb{R}^n$ known as the \emph{Euclidean embedded graph} and illustrate how to define the associated \emph{mass-action system} based on this graph structure.

\begin{definition}[\cite{craciun2015toric,craciun2020endotactic}]

\begin{enumerate}
\item[(a)] A \defi{reaction network} $G = (V, E)$, also known as a \defi{Euclidean embedded graph (or E-graph)}, is a directed graph in $\RR^n$, where $V \subset \mathbb{R}^n$ denotes a finite set of \defi{vertices}, and $E \subseteq V \times V$ denotes a finite set of \defi{edges}.
In this context, we assume that there are no isolated vertices, no self-loops, and at most one edge between any pair of ordered vertices.

\item[(b)] Let $V = \{ \by_1, \ldots, \by_m \}$.
A directed edge $(\by_i, \by_j) \in E$, referred to as a \defi{reaction} in the network, is also denoted by $\by_i \to \by_j$, where $\by_i$ and $\by_j$ are termed the \defi{source vertex} and \defi{target vertex}, respectively.
Furthermore, we define the \defi{reaction vector} associated with the reaction $\by_i \rightarrow \by_j$ to be $\by_j - \by_i \in \mathbb{R}^n$. 
\end{enumerate}
\end{definition}

\begin{definition} 

Let $G=(V, E)$ be an E-graph.
\begin{enumerate}
\item[(a)] The set of vertices $V$ is partitioned into its connected components, which are also referred to as \defi{linkage classes}.

\item[(b)] A connected component $L \subseteq V$ is termed \defi{strongly connected} if every edge in $L$ is part of a directed cycle within $L$.
Moreover, $G = (V, E)$ is termed \defi{weakly reversible} if each of its connected components is strongly connected.
\end{enumerate}
\end{definition}

\begin{example} 
\label{ex:complexesSpecies} 

Figure~\ref{fig:e-graph} depicts an example of a reaction network represented as an E-graph $G = (V, E)$.
There are three vertices:
\begin{equation} \notag
V = \{ \by_1, \ \by_2, \ \by_3 \},
\end{equation}
and four directed edges between vertices:
\begin{equation} \notag
E = \{ \by_1 \to \by_2, \ \by_2 \to \by_1, \ \by_2 \to \by_3, \ \by_3 \to \by_1 \}.
\end{equation}
 
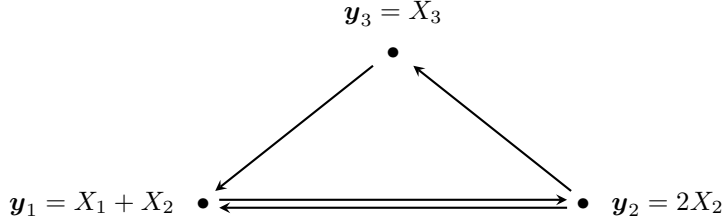
\begin{figure}[H] 
\begin{center}
\begin{tikzpicture}
    \node (1) at (-2.5,0) {$\bullet$};
    \node (3) at (0,2) {$\bullet$};
    \node (2) at (2.5,0) {$\bullet$};
    \node [left=1pt of 1] {$\by_1 = X_1 + X_2$};
    \node [right=1pt of 2] {$\by_2 = 2 X_2$};
    \node [above=1pt of 3] {$\by_3 = X_3$};
	\draw [{->}, -{stealth}, thick, transform canvas={yshift=1.5pt}] (1) -- (2) node [midway, above] {};
	\draw [{->}, -{stealth}, thick, transform canvas={yshift=-1.5pt}] (2) -- (1) node [midway, below] {};
	\draw [{->}, -{stealth}, thick, transform canvas={xshift=1.5pt}] (2) -- (3) node [midway, right=6pt] {};
	\draw [{->}, -{stealth}, thick, transform canvas={xshift=-1.5pt}] (3) -- (1) node [midway, left] {};
\end{tikzpicture}
\end{center}
\caption{This E-graph $G = (V, E)$ is weakly reversible and consists of a single linkage class.}
\label{fig:e-graph}
\end{figure} 
\qed
\end{example} 

\begin{definition}[\cite{feinberg1979lectures}]
\label{def:mass_action}

Let $G=(V, E)$ be an E-graph. 
Define a \defi{reaction rate vector} $\bk$ as
\[
\bk = (k_{\by_i \rightarrow \by_j})_{\by_i \rightarrow \by_j \in E} \in \mathbb{R}^{|E|}_{>0},
\]
where $k_{\by_i \rightarrow \by_j}$ or $k_{ij}$ is the \defi{reaction rate constant} associated with the edge $\by_i \rightarrow \by_j \in E$. The \defi{associated mass-action dynamical system} generated by $(G, \bk)$ is a dynamical system on $\RR_{>0}^n$ defined by
\begin{equation} \label{eq:mass_action}
\frac{d \bx}{d t} 
= \sum_{\by_i \rightarrow \by_j \in E}k_{\by_i \rightarrow \by_j} \bx^{\by_i}(\by_j - \by_i),
\end{equation}
where $\bx^{\by} := x_1^{y_{1}} x_2^{y_{2}} \ldots x_n^{y_{n}}$.\footnote{Note that when $\by = \emptyset$, we have $\bx^{\by} = x_1^{0} \ldots x_n^{0}$ = 1.}
Furthermore, we define the \defi{stoichiometric subspace} of $G$ as the span of the reaction vectors of $G$, represented by
\begin{equation} \notag
\mS = \spn \{ \by_j - \by_i: \by_i \rightarrow \by_j \in E \}.
\end{equation}
\end{definition}

In \cite{sontag2001structure}, it is shown that the positive orthant $\RR_{>0}^n$ is forward-invariant when $V \subset \ZZ_{\geq 0}^n$.
Therefore, in this context we always assume $V \subset \ZZ_{\geq 0}^n$, ensuring that any solution to \eqref{eq:mass_action} with initial condition $\bx_0 \in \RR_{>0}^n$ and $V \subset \ZZ_{\geq 0}^n$ is confined to the set $(\bx_0 + \mS) \cap \RR_{>0}^n$, which defines the \defi{invariant polyhedron} of $G$ at $\bx_0$.

\begin{example}
\label{ex:intro} 

The E-graph $G = (V, E)$ in \Cref{fig:mass-action} has 3 edges and 6 vertices.
Given the reaction rate constants shown in \Cref{fig:mass-action}, the associated mass-action dynamical system $(G, \bk)$ is
\begin{equation} \notag
\begin{split}
\frac{d \bx}{d t} 
&= x_1 \begin{pmatrix} 1 \\ 0 \end{pmatrix} + x_1 x_2 \begin{pmatrix} -1 \\ 1 \end{pmatrix} + x_2 \begin{pmatrix} 0 \\ -1 \end{pmatrix} 
= \begin{pmatrix} x_1 - x_1 x_2 \\[5pt] x_1 x_2 - x_2 \end{pmatrix}. 
\end{split}
\end{equation}
This is the Lotka--Volterra population dynamics model \cite{Lotka--Volterra}.
\begin{figure}[!ht]
\centering 
\begin{tikzpicture}
\draw [opacity=0, -{stealth}, thick, blue, transform canvas={ yshift=-0.31ex}] (0,0)--(2,0) node [midway, below] {\footnotesize $1$};
\draw [step=1, gray!50!white, thin] (0,0) grid (2.5,2.5);
\node at (0,2.75) {};
\node at (0,-0.25) {};
\draw [->, gray] (0,0)--(2.5,0);
\draw [->, gray] (0,0)--(0,2.5);
\node [inner sep=0pt, blue] (y) at (0,1) {$\bullet$};
\node [inner sep=0pt, blue] (0) at (0,0) {$\bullet$};
\node [inner sep=0pt, blue] (xy) at (1,1) {$\bullet$};
\node [inner sep=0pt, blue] (2y) at (0,2) {$\bullet$};
\node [inner sep=0pt, blue] (x) at (1,0) {$\bullet$};
\node [inner sep=0pt, blue] (2x) at (2,0) {$\bullet$};
\draw [-{stealth}, thick, blue] (y)--(0) node [near start, right] {\ratecnst{$1$}};
\draw [-{stealth}, thick, blue] (xy)--(2y) node [near start, above] {\ratecnst{$1$}};
\draw [-{stealth}, thick, blue] (x)--(2x) node [near start, above] {\ratecnst{$1$}};
\end{tikzpicture}
\caption{The E-graph $G$ consists of $6$ vertices and $3$ edges. Under mass-action kinetics, this network corresponds to the classical Lotka–Volterra model for population dynamics.}
\label{fig:mass-action}
\end{figure}
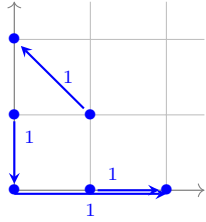
\qed
\end{example} 

\subsection{Complex-Balanced Systems}

In general, mass-action systems exhibit diverse dynamics. In this section, we introduce a special class of systems: \emph{complex-balanced dynamical systems}.
The significance of complex-balanced dynamical systems derives from the strong stability properties of their positive equilibria, known as the complex-balanced equilibria.

\begin{definition} 
\label{def:CB}

Let $(G, \bk)$ be a mass-action system as follows:
\begin{equation} \notag
 \frac{d \bx}{d t} 
= \sum_{\by\rightarrow \by' \in E}k_{\by \rightarrow \by'} \bx^{\by}(\by' - \by).
\end{equation}
A point $\bx^* \in \mathbb{R}_{>0}^n$ is termed a \defi{positive equilibrium} of $(G, \bk)$ if it satisfies
\[ 
\sum_{\by \rightarrow \by' \in E}k_{\by \rightarrow \by'} (\bx^*)^{\by}(\by' - \by)
= \mathbf{0}.
\]
A positive equilibrium $\bx^* \in \RR_{>0}^n$ is called a \defi{complex-balanced equilibrium} of $(G, \bk)$ if for every vertex $\by \in V$,
\begin{equation} \notag
\sum_{\by \to \hat{\by} \in E} k_{\by \to \hat{\by}} (\bx^*)^{\by}
= \sum_{\by' \to \by \in E}k_{\by' \to \by}
(\bx^*)^{\by'}.
\end{equation}
We say that the pair $(G, \bk)$ \defi{satisfies the complex-balanced conditions} if it has a complex-balanced equilibrium. In this case, the mass-action system generated by $(G, \bk)$ is referred to as a \defi{complex-balanced system}. 
\end{definition}

Complex-balanced systems exhibit both algebraic and dynamical properties. The following classical theorem highlights key properties of complex-balanced systems.

\begin{theorem}[\cite{HornJackson1972}]
\label{thm:HJ}

Let $(G, \bk)$ be a complex-balanced system with a positive equilibrium $\bx^* \in \mathbb{R}_{>0}^n$.
Then we have the following:

\begin{enumerate}
\item[(a)] All positive equilibria of complex-balanced systems are complex-balanced equilibria, and there is exactly one equilibrium within each invariant polyhedron of the system.

\item[(b)] Every complex-balanced equilibrium $\bx$ satisfies
\[
\ln \bx - \ln \bx^* \in \mathcal{S}^\perp,
\]
where $\mathcal{S}$ is the stoichiometric subspace of $G$, and $\ln(\cdot)$ is defined component-wise on vectors and element-wise on sets of vectors.
 
\item[(c)] Every complex-balanced equilibrium is asymptotically stable within its invariant polyhedron.
\end{enumerate}
\qed
\end{theorem}

The necessary and sufficient conditions for the system $(G, \bk)$ to be complex-balanced are:
(1) the graph $G = (V, E)$ is weakly reversible, and (2) the reaction rate vector $\bk$ satisfies a specific set of algebraic equations, the number of which is determined by a non-negative integer known as the \emph{deficiency}. 

\begin{definition}
\label{def:deficiency}

Let $G = (V, E)$ be an E-graph with $\ell$ linkage classes, and let $S$ be its associated stoichiometric subspace. The \defi{deficiency} of $G = (V, E)$ is an integer defined as
\[
\delta = |V| - \ell - \dim \mS.
\]
\end{definition}

Under mass-action kinetics, networks with low deficiency exhibit special dynamical properties. Specifically, the Deficiency Zero Theorem states that weakly reversible networks with a deficiency of zero are complex-balanced for any choice of rate constants \cite{horn1972necessary}.

\begin{theorem}[Deficiency Zero Theorem, \cite{horn1972necessary}]
\label{thm:deficiency_zero}

Let $G = (V, E)$ be an E-graph. 
The mass-action system $(G, \bk)$ is complex-balanced for any choice of reaction rate vector $\bk \in \mathbb{R}^{|E|}_{>0}$ if and only if the E-graph $G$ is weakly reversible and has zero deficiency.
\qed
\end{theorem}

Consider Example \ref{ex:complexesSpecies}, the E-graph $G$ from Figure~\ref{fig:e-graph} is weakly reversible and consists of three vertices $(|V| = 3)$ in one linkage class $(\ell=1)$. Additionally, the dimension of the stoichiometric subspace is two $(\dim \mS = 2)$. Thus, the deficiency is $\delta = 3 - 1 - 2 = 0$.
From Theorem \ref{thm:deficiency_zero}, the mass-action system $(G, \bk)$ is always complex-balanced for any choice of reaction rate vector $\bk \in \mathbb{R}^{|E|}_{>0}$.

\begin{definition}
\label{def:conservation_law}

Let $G = (V, E)$ be an E-graph. A function $\phi (\bx)$ is termed a \defi{conservation law} if it is a first integral of the system \eqref{eq:mass_action} for every reaction rate vector $\bk$, that is,
\[
\sum\limits^{n}_{i=1} \frac{\partial \phi}{\partial x_i} (\bx) \frac{d x_i}{d t} = 0.
\]
\end{definition}

Consider the mass-action system $(G, \bk)$ given by \eqref{eq:mass_action}, and let $\mathcal{S}$ be the associated stoichiometric subspace.
For any vector $\bu = (u_1, \ldots, u_n)^{\intercal} \in \mathcal{S}^\perp$, we have
\be \label{eq:linear_conservation}
\sum\limits^{n}_{i=1} u_i \frac{d x_i}{d t} = 0,
\ee
and thus $\sum\limits^{n}_{i=1} u_i x_i$ is a linear conservation law.
For simplicity, we let $\bu = (u_1, \ldots, u_n)^{\intercal}$ denote the linear conservation law $\sum\limits^{n}_{i=1} u_i x_i$.

\begin{example}

Revisiting Example \ref{ex:complexesSpecies}, the stoichiometric subspace $\mathcal{S}$ of the E-graph $G$ from Figure~\ref{fig:e-graph} is given by
\[
\mS = \spn \{ \by_2 - \by_1, \by_1 - \by_2, \by_1 - \by_3, \by_3 - \by_2 \} = \spn \{ (1, -1, 0)^{\intercal}, (0,-1,1)^{\intercal} \}.
\]
Then we obtain that $(1, 1, 1)^{\intercal} \in \mathcal{S}^\perp$ and derive a conservation law
\[
\phi(x) = x_1 + x_2 + x_3.
\]
\qed
\end{example}

\subsection{Infinitesimal Homeostasis}
\label{ss:inf_homeo}

In this section, we continue to introduce some key concepts, including \emph{input-output systems}, \emph{input-output functions}, and \emph{infinitesimal homeostasis}. 
Some of our discussion is based on \cite{GS17, WHAG21}.  

\begin{definition} 
\label{defn:io_system}

Let $\bx = (x_1, x_2, \ldots, x_n) \in \mathbb{R}^n$ be the vector of state variables, and let $\II \in \mathbb{R}$ be the \defi{input parameter}. 
The \defi{input-output system} is a dynamical system on $\RR_{>0}^n$ defined by
\begin{equation}
\begin{split} \label{eq:io_system}
\frac{d x_1}{d t} & = f_1 (x_1, \ldots, x_n, \II),
\\ \frac{d x_2}{d t} & = f_2 (x_1, \ldots, x_n), 
\\ & \ \vdots 
\\ \frac{d x_n}{d t} & = f_n (x_1, \ldots, x_n),
\end{split}
\end{equation}
where $\bf (\bx, \II) = (f_1 (\bx, \II), \ldots, f_n (\bx))$ is a smooth family of mappings on the state space $\mathbb{R}^n$.
In this context, we always assume that the input parameter $\II$ appears only in the first equation of \eqref{eq:io_system} and the state variable $x_n$ is the \defi{output parameter}. 
\end{definition}

The input-output system \eqref{eq:io_system} can be written as 
\begin{equation} \label{eq2:io_system}
\frac{d \bx}{d t} = \bf (\bx, \II).
\end{equation}
In \cite{WHAG21}, it is shown that when the system $\frac{d \bx}{d t} = \bf (X, \II_0)$ has a \emph{hyperbolic} equilibrium $\bx = \bx_0$, then the implicit function theorem implies the existence of a unique family of equilibria
\[
\bx (\II) = (x_1 (\II), \ldots, x_n (\II)),
\]
such that $\bx (\II_0) = \bx_0$ and $\bf (\bx (\II), \II) = 0$ for every $\II$ near $\II_0$.
With the function $x_n = x_n(\II)$, we can illustrate the relationship between the output parameter $x_n$ and the input parameter $\II$. 

\emph{Homeostasis} occurs in a system of differential equations when the output parameter $x_n$ remains approximately constant while the input parameter $\II$ varies. 
In \cite{GS17}, the concept of \emph{infinitesimal homeostasis} is introduced as follows.

\begin{definition}[\cite{GS17}]
\label{def:inf_homeo}

Given the input-output system \eqref{eq:io_system}, suppose the mapping $\II \mapsto x_n (\II)$ is well-defined in a neighborhood of $\II_0$, then $x_n (\II)$ is called the \defi{input-output function}.
Moreover, \defi{infinitesimal homeostasis} occurs at $\II_0$ if 
\[
\frac{d}{d \II} x_n (\II_0) = 0.
\] 
Since a function varies slowly near a stationary point, infinitesimal homeostasis is a \emph{sufficient} condition for homeostasis over some interval of the input parameter~\cite{GS17}.
\end{definition} 

In \cite{WHAG21}, the authors proved that infinitesimal homeostasis occurs when the determinant of the homeostasis matrix $H$ is $0$. In this context, the \defi{homeostasis matrix} $H$ is obtained from the Jacobian matrix $J$ of \eqref{eq:io_system} by deleting its row corresponding to $x_1$ and its column corresponding to $x_n$. More preciously,
\be \label{def:J_and_H}
J = \begin{pmatrix}
f_{1, 1} & \ldots & f_{1, n-1} & f_{1, n} \\
f_{2, 1} & \ldots & f_{2, n-1} & f_{2, n} \\
\vdots \ \ \ & & & \vdots \ \ \ \\
f_{n, 1} & \ldots & f_{n, n-1} & f_{n, n} 
\end{pmatrix}
\qquad \Longrightarrow \qquad 
H = \begin{pmatrix}
f_{2, 1} & \ldots & f_{2, n-1} \\
\vdots \ \ \ &  & \vdots \ \ \ \ \\
f_{n, 1} & \ldots & f_{n, n-1}
\end{pmatrix},
\ee
where $f_{j,\ell}$ denotes the partial derivative of the function $f_j$ with respect to the state variable $x_\ell$.

\begin{theorem}[{\cite[Theorem 1.11]{WHAG21}}] \label{thm:irreducible}

Assume the system \eqref{eq:io_system} has a hyperbolic equilibrium at $(\bx_0, \II_0)$. 
There are permutation matrices $P$ and $Q$ such that $PHQ$ is block upper triangular with irreducible square blocks $B_1, \ldots, B_m$, that is,
\[
PHQ = 
\begin{pmatrix}
B_1 & \cdots & * & \cdots & * \\
 & \ddots &  & & \\
 & & B_{\eta} & \cdots & * \\
 & & & \ddots & \\
&  & &  & B_{m} 
\end{pmatrix}.
\]
Then
\begin{equation} \label{eq:irreducible}
\det(H) = \det(B_1) \cdots \det(B_m).
\end{equation}
\qed
\end{theorem}

From Theorem \ref{thm:irreducible}, computing the determinant of the homeostasis matrix reduces to computing the determinant of these irreducible blocks.

\begin{definition}[\cite{WHAG21}]
\label{def:homeo_type}

Each irreducible square block $B_{\eta}$ in \eqref{eq:irreducible} is referred to as a \defi{homeostasis block}. 
Moreover, an infinitesimal homeostasis is said to be of \defi{homeostasis type} $B_{\eta}$ if
\begin{equation} \notag
\det (B_{\eta}) = 0 \ \text{ and } \ \det (B_{\xi}) \neq 0
\ \text{ for all } \ \xi \neq \eta.
\end{equation}
\end{definition}


\section{Basic Properties}
\label{sec:basic}

In this section, we briefly review  some basic properties of the  mathematical theory of  homeostasis as applied to  (deterministic) mass-action systems and relevant to our resulst in the subsequent sections. 
We begin by considering the mass-action systems whose stoichiometric subspaces encompass  the entire space $\mathbb{R}^n$.

\subsection{Homeostasis in Mass-action Systems with \texorpdfstring{$\mS = \mathbb{R}^n$}{S=Rn}}
\label{ss:homeo_mas=Rn}

To examine homeostasis in mass-action systems, we need to construct an input-output system based on Definition \ref{defn:io_system} as follows:
\begin{equation} 
\begin{split} \label{eq3:io_system}
\frac{d x_1}{d t} & = f_1 (x_1, \ldots, x_n, \II),
\\ \frac{d x_2}{d t} & = f_2 (x_1, \ldots, x_n), 
\\ & \ \vdots 
\\ \frac{d x_n}{d t} & = f_n (x_1, \ldots, x_n),
\end{split}
\end{equation}
where $f_1, \ldots, f_n$ are smooth mappings satisfying the mass-action kinetics in Definition \ref{def:mass_action}.

Recall that we always assume the input parameter $\II$ appears only in the first equation of \eqref{eq3:io_system}.
Thus, according to the mass-action kinetics in Definition \ref{def:mass_action}, we consider an enzyme reaction with the reaction rate constant as the input parameter, given by
\begin{equation} \label{def:enzyme_reaction}
\emptyset \xrightarrow{\II} a X_1,
\end{equation}
where $a \in \mathbb{Z}_{>0}$, and $x_n$ is the output parameter.

Suppose that the system \eqref{eq3:io_system} has a hyperbolic equilibrium $\bx = \bx_0$ when $\II = \II_0$, and the input-output function $x_n (\II_0)$ is well-defined in a neighborhood of $\II_0$.
From Definition \ref{def:inf_homeo}, along with \cite{WHAG21}, infinitesimal homeostasis occurs at $\II_0$ when the determinant of the homeostasis matrix $H$ defined in \eqref{def:J_and_H} is $0$.

\begin{example}
\label{ex1:homeo_mas}

Consider the E-graph $G = (V, E)$ in Figure~\ref{fig:input_ouput_mas}. Suppose all reaction rate constants are equal to 1 except for the rate of the enzyme reaction, which is $\II$, represented as
\[
\emptyset \xrightarrow{\II} X_1.
\]
The mass-action system is then given by:
\begin{equation}
\begin{split} \label{eq:ex1:homeo_mas}
& \frac{d x_1}{dt} = - 2 x_1 + x_2 + \II,
\\& \frac{d x_2}{dt} = x_1 - (x_2)^3, 
\\& \frac{d x_3}{dt} = x_2 + (x_2)^3 - 2 (x_2)^2 x_3. 
\end{split}
\end{equation}
Let $\bx^* = (x^*_1, x^*_2, x^*_3)$ be the equilibrium with corresponding $\II = \II^*$, then we have
\[
2 x^*_1 = x^*_2 + \II^*, \
x^*_1 = (x^*_2)^3, \
x^*_2 + (x^*_2)^3 = 2 (x^*_2)^2 x^*_3.
\]
Moreover, we obtain the homeostasis matrix as follows: 
\begin{equation} \notag
H = 
\begin{pmatrix} 
1 & - 3 (x_2)^2 \\ 
0 & 1 + 3 (x_2)^2 - 4 x_2 x_3 
\end{pmatrix},
\end{equation}
and thus infinitesimal homeostasis occurs at $\bx^* = (x^*_1, x^*_2, x^*_3)$ when
\[
\det (H) = 1 + 3 (x^*_2)^2 - 4 x^*_2 x^*_3 = 0.
\]

We consider the equilibrium $(x^*_1, x^*_2, x^*_3) = (1, 1, 1)$ which satisfies the above equations with corresponding $\II^* = 1$.
We further verify that the Jacobian matrix is non-degenerate at this point, indicating it is a hyperbolic equilibrium. Therefore, we conclude that infinitesimal homeostasis occurs at $\II^* = 1$.

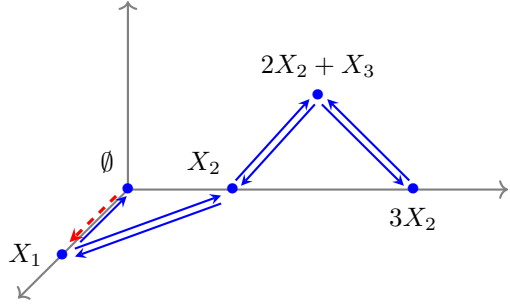
\begin{figure}[!ht]
\centering
\begin{tikzpicture}[scale=1.25]
    \draw [->, thick, gray] (0,0,0)--(4,0,0) node[right]{};
    \draw [->, thick, gray] (0,0,0)--(0,2,0) node[above]{};
    \draw [->, thick, gray] (0,0,0)--(0,0,3) node[below]{};
    \node [inner sep=0pt, blue] (1) at (1,0,0) {};
    \node [inner sep=0pt, blue] (12) at (1.1,0,0) {$\bullet$};
    \node [inner sep=0pt, blue] (11) at (0.8,0,0) {};
    \node [above=1.5pt of 11] {$X_2$};
    \node [inner sep=0pt, blue] (2) at (2,1,0) {$\bullet$};
    \node [above=1pt of 2] {$2 X_2 + X_3$};
    \node [inner sep=0pt, blue] (3) at (3,0,0) {$\bullet$};
    \node [below=1pt of 3] {$3 X_2$};
    \node [inner sep=0pt, blue] (4) at (0,0,1.5) {};
    \node [inner sep=0pt, blue] (41) at (0,0,1.8) {$\bullet$};
    \node [left=1pt of 41] {$X_1$};
    \node [inner sep=0pt, blue] (5) at (0,0,0) {$\bullet$};
    \node [inner sep=0pt, blue] (51) at (0,0.3,0) {};
    \node [left=1pt of 51] {$\emptyset$};
    \draw [{->}, -{stealth}, red, line width=1.2pt, dashed, transform canvas={xshift=-2pt}] (5) -- (4) node [midway, above] {};
    \draw [{->}, -{stealth}, blue, thick, transform canvas={xshift=2pt}] (4) -- (5) node [midway, below] {};
    \draw [{->}, -{stealth}, blue, thick, transform canvas={yshift=-2pt}] (4) -- (1) node [midway, above] {};
    \draw [{->}, -{stealth}, blue, thick, transform canvas={yshift=-5pt}] (1) -- (4) node [midway, below] {};
    \draw [{->}, -{stealth}, blue, thick, transform canvas={xshift=1pt, yshift=-1pt}] (2) -- (12) node [midway, above] {};
    \draw [{->}, -{stealth}, blue, thick, transform canvas={xshift=-1pt, yshift=1pt}] (12) -- (2) node [midway, below] {};
    \draw [{->}, -{stealth}, blue, thick, transform canvas={xshift=-1pt, yshift=-1pt}] (2) -- (3) node [midway, above] {};
    \draw [{->}, -{stealth}, blue, thick, transform canvas={xshift=1pt, yshift=1pt}] (3) -- (2) node [midway, below] {};
\end{tikzpicture}
\caption{The E-graph $G = (V, E)$ contains an enzyme reaction with the input parameter (reaction rate constant) $\II$ depicted by the red dashed arrow.}
\label{fig:input_ouput_mas}
\end{figure}
\qed
\end{example}

\subsection{Infinitesimal Homeostasis in Complex-Balanced Systems}
\label{ss:homeo_cb}

In Sections \ref{ss:inf_homeo} and \ref{ss:homeo_mas=Rn}, both the input-output function $x_n'(\II_0)$ and infinitesimal homeostasis are well-defined around a hyperbolic equilibrium when $\II = \II_0$. On the other hand, given any complex-balanced system $(G, \bk)$ with $\mS = \mathbb{R}^n$, Theorem \ref{thm:HJ} shows that every complex-balanced equilibrium is asymptotically stable, and thus it is hyperbolic.

The above fact brings up the following question:
\emph{Given a complex-balanced system $(G, \bk)$ with an input parameter $\II$, when does infinitesimal homeostasis occurs for some $\II = \II_0$?}

\medskip

We begin by formulating an input-output system \eqref{eq3:io_system}, where we impose certain constraints on $(G, \bk)$ and the input parameter $\II$ as follows:
\begin{enumerate}

\item[(a)] $G$ is weakly reversible and the associated stoichiometric subspace $\mS = \mathbb{R}^n$.

\item[(b)] $G$ contains only one enzyme reaction with the reaction rate constant as the input parameter, given by
\begin{equation} \label{eq:enzyme}
\emptyset \xrightarrow{\II} a X_1
\ \text{ with } \
a \in \mathbb{Z}_{>0}.
\end{equation}
Hence, for any other reaction $\by \to \by' \in E$, the source vertex $\by \neq \emptyset$. 

\item[(c)] For every reaction $\by \to \by' \in E$ except the enzyme reaction \eqref{eq:enzyme}, it has a fixed reaction rate constant $k_{\by \to \by'}$. 
\end{enumerate}
For simplicity, we denote this input-output system by $(G, \bk, \II)$.
To address the above question regarding $(G, \bk, \II)$, we construct an associated mass-action system $(\hat{G}, \hbk)$ as follows:

\medskip

\textbf{Step 1. } Let $\hat{V} = V \ \backslash \ \{ \emptyset \}$. For every reaction $\by \xrightarrow{k} \by' \in E$ where $\by \neq \emptyset, \by' \neq \emptyset$, we add
\[
\by \xrightarrow{k} \by' \in \hat{E}.
\]

\textbf{Step 2. }
For every reaction $\by \xrightarrow{k} \emptyset \in E$ with $\by \neq a X_1$, we replace this reaction in $\hat{G}$ as
\[
\by \xrightarrow{k} a X_1 \in \hat{E}.
\]
\qed

\medskip

Now we demonstrate that it suffices to examine the mass-action system $(\hat{G}, \hbk)$ to determine whether the system $(G, \bk, \II)$ exhibits infinitesimal homeostasis at $\II_0$, with $(G, \bk, \II_0)$ being complex-balanced.

\begin{proposition}
\label{prop:hatGk_GkI}

Let $(\hat{G}, \hbk)$ be a complex-balanced system with a positive equilibrium $\bx^* \in \mathbb{R}_{>0}^n$.
Suppose its homeostasis matrix $\hat{H}$ satisfies $\det (\hat{H}) = 0$ when $\bx = \bx^*$. 
There exists $\II_0 > 0$ such that the input-output system $(G, \bk, \II)$ exhibits infinitesimal homeostasis at $\II_0$ and $(G, \bk, \II_0)$ is a complex-balanced system.
\end{proposition}

\begin{proof}

For simplicity, we denote the enzyme reaction in $G$ by
\[
\emptyset \xrightarrow{\II} \by_0 := a X_1.
\]

First, we show $(G, \bk, \II_0)$ is a complex-balanced system.
From the assumption, $(\hat{G}, \hbk)$ has a complex-balanced equilibrium $\bx^* \in \mathbb{R}_{>0}^n$, that is,
\begin{equation} \label{eq:hatGk_GkI}
\sum_{\by_0 \to \by \in \hat{E}} \hat{k}_{\by_0 \to \by} (\bx^*)^{\by_0}
= \sum_{\by' \to \by_0 \in \hat{E}} \hat{k}_{\by' \to \by_0}
(\bx^*)^{\by'}.
\end{equation}
Since $G$ is weakly reversible, there are some reactions 
\begin{equation} \label{eq3:hatGk_GkI}
\by \xrightarrow{k} \emptyset \in E.
\end{equation}
Under this construction, for $\by \neq \by_0$, the reactions in \eqref{eq3:hatGk_GkI} are replaced in $\hat{E}$ as follows:
\[
\by \xrightarrow{k} \by_0 \in \hat{E}
\ \text{ with } \
\by \neq \by_0,
\]
and thus for every reaction $\by \to \by_0 \in \hat{E}$ with $\by \neq \by_0$, we have
\begin{equation} \label{eq2:hatGk_GkI}
\hat{k}_{\by \to \by_0} = k_{\by \to \by_0} + k_{\by \to \emptyset}.
\end{equation}
Let $\II_0 = \sum_{\by \to \emptyset \in E} k_{\by \to \emptyset} (\bx^*)^{\by}$, we claim that the input-output system $(G, \bk, \II_0)$ has a complex-balanced equilibrium $\bx^*$.

For the vertex $\emptyset$, since there is only one enzyme reaction $\emptyset \to \by_0 \in E$, then
\[
\sum\limits_{\emptyset \to \by \in E} k_{\emptyset \to \by} = \II_0 = \sum\limits_{\by \to \emptyset \in E} k_{\by \to \emptyset} (\bx^*)^{\by}.
\]
For the vertex $\by_0$, from the setting of $\II_0$ and \eqref{eq2:hatGk_GkI} we have
\begin{equation}
\begin{split} \notag
\sum_{\by' \to \by_0 \in E} k_{\by' \to \by_0}
(\bx^*)^{\by'} 
& = \II_0 + \sum_{\by' \to \by_0 \in E, \ \by' \neq \emptyset} k_{\by' \to \by_0}
(\bx^*)^{\by'}
\\& = \sum\limits_{\by \to \emptyset \in E} k_{\by \to \emptyset} (\bx^*)^{\by} + \sum_{\by' \to \by_0 \in E, \ \by' \neq \emptyset} k_{\by' \to \by_0}
(\bx^*)^{\by'} 
\\& = k_{\by_0 \to \emptyset} (\bx^*)^{\by_0} + \sum_{\by' \to \by_0 \in \hat{E}} \hat{k}_{\by' \to \by_0} (\bx^*)^{\by'}.
\end{split}
\end{equation}
Together with \eqref{eq:hatGk_GkI}, we obtain that
\begin{equation}
\begin{split} \notag
& k_{\by_0 \to \emptyset} (\bx^*)^{\by_0} + \sum_{\by' \to \by_0 \in \hat{E}} \hat{k}_{\by' \to \by_0} (\bx^*)^{\by'} 
\\& = k_{\by_0 \to \emptyset} (\bx^*)^{\by_0} +  \sum_{\by_0 \to \by \in \hat{E}} \hat{k}_{\by_0 \to \by} (\bx^*)^{\by_0}
= \sum_{\by_0 \to \by \in E} k_{\by_0 \to \by} (\bx^*)^{\by_0}.
\end{split}
\end{equation}
For the remaining vertices in $G$, we can check that they satisfy the complex-balanced conditions due to \eqref{eq2:hatGk_GkI}. Therefore, we prove the claim and thus $(G, \bk, \II_0)$ is a complex-balanced system.

\smallskip

Second, we show that the input-output system $(G, \bk, \II_0)$ exhibits infinitesimal homeostasis.
From the construction, both the system $(G, \bk, \II)$ and the mass-action system $(\hat{G}, \hbk)$ share the same right-hand side except for the term involving $\frac{d x_1}{d t}$.

Moreover, the homeostasis matrix is derived from the Jacobian matrix by removing its row corresponding to $x_1$ and its column corresponding to $x_n$.
Hence, we conclude that the homeostasis matrix $H$ for the system $(G, \bk, \II)$ satisfies
\begin{equation} \label{eq4:hatGk_GkI}
\det (H) = \det (\hat{H}) = 0
\ \text{ at } \
\bx = \bx^*.
\end{equation}
In the first part, we have shown that $\bx^*$ is a complex-balanced equilibrium of the input-output system $(G, \bk, \II_0)$ and it is hyperbolic since the stoichiometric subspace of $G$ is $\mS = \mathbb{R}^n$. Together with \eqref{eq4:hatGk_GkI}, we conclude that $(G, \bk, \II)$ exhibits infinitesimal homeostasis at $\II_0$.
\end{proof}

\begin{remark}
Compared to the system $(G, \bk, \II)$, the mass-action system $(\hat{G}, \hbk)$ does not include the vertex $\emptyset$ or the input parameter $\II$. Consequently, it is simpler to verify whether $(\hat{G}, \hbk)$ is a complex-balanced system and to compute its homeostasis matrix.
Therefore, Proposition \ref{prop:hatGk_GkI} offers an efficient method for verifying whether an input-output system $(G, \bk, \II)$ exhibits infinitesimal homeostasis at $\II_0$ while the system is complex-balanced.
\end{remark}

\begin{example}
\label{ex1:homeo_cb}

We revisit the input-output system $(G, \bk, \II)$ in Example \ref{ex1:homeo_mas}, where $G$ represents the E-graph in Figure~\ref{fig:input_ouput_mas}, and all reaction rate constants are equal to 1, except for the rate of the enzyme reaction $\II$.

Then we construct the associated mass-action system $(\hat{G}, \hbk)$ in Figure~\ref{fig2:input_ouput_mas} and it satisfies
\begin{equation}
\begin{split} \label{eq:ex1:homeo_cb}
& \frac{d x_1}{dt} = - x_1 + x_2,
\\& \frac{d x_2}{dt} = x_1 - (x_2)^3, 
\\& \frac{d x_3}{dt} = x_2 + (x_2)^3 - 2 (x_2)^2 x_3. 
\end{split}
\end{equation}
By direct computation, we verify that $\bx^* = (1, 1, 1)$ is a complex-balanced equilibrium.
Moreover, the homeostasis matrix of $(\hat{G}, \hbk)$ is given by 
\begin{equation} \notag
H = 
\begin{pmatrix} 
1 & - 3 (x_2)^2 \\ 
0 & 1 + 3 (x_2)^2 - 4 x_2 x_3 
\end{pmatrix},
\end{equation}
and thus
\[
\det (H) = 1 + 3 (x^*_2)^2 - 4 x^*_2 x^*_3 = 0
\ \text{ at } \
\bx^* = (1, 1, 1).
\]

By setting $\II_0 = 1$, Proposition \ref{prop:hatGk_GkI} shows that the input-output system $(G, \bk, \II)$ exhibits infinitesimal homeostasis at $\II_0$ and $(G, \bk, \II_0)$ is a complex-balanced system.

\begin{figure}[!ht]
\centering
\begin{tikzpicture}[scale=1.25]
    \draw [->, thick, gray] (0,0,0)--(4,0,0) node[right]{};
    \draw [->, thick, gray] (0,0,0)--(0,2,0) node[above]{};
    \draw [->, thick, gray] (0,0,0)--(0,0,3) node[below]{};
    \node [inner sep=0pt, blue] (1) at (1,0,0) {};
    \node [inner sep=0pt, blue] (12) at (1.1,0,0) {$\bullet$};
    \node [inner sep=0pt, blue] (11) at (0.8,0,0) {};
    \node [above=1.5pt of 11] {$X_2$};
    \node [inner sep=0pt, blue] (2) at (2,1,0) {$\bullet$};
    \node [above=1pt of 2] {$2 X_2 + X_3$};
    \node [inner sep=0pt, blue] (3) at (3,0,0) {$\bullet$};
    \node [below=1pt of 3] {$3 X_2$};
    \node [inner sep=0pt, blue] (4) at (0,0,1.5) {};
    \node [inner sep=0pt, blue] (41) at (0,0,1.8) {$\bullet$};
    \node [left=1pt of 41] {$X_1$};
    \draw [{->}, -{stealth}, blue, thick, transform canvas={yshift=-2pt}] (4) -- (1) node [midway, above] {};
    \draw [{->}, -{stealth}, blue, thick, transform canvas={yshift=-5pt}] (1) -- (4) node [midway, below] {};
    \draw [{->}, -{stealth}, blue, thick, transform canvas={xshift=1pt, yshift=-1pt}] (2) -- (12) node [midway, above] {};
    \draw [{->}, -{stealth}, blue, thick, transform canvas={xshift=-1pt, yshift=1pt}] (12) -- (2) node [midway, below] {};
    \draw [{->}, -{stealth}, blue, thick, transform canvas={xshift=-1pt, yshift=-1pt}] (2) -- (3) node [midway, above] {};
    \draw [{->}, -{stealth}, blue, thick, transform canvas={xshift=1pt, yshift=1pt}] (3) -- (2) node [midway, below] {};
\end{tikzpicture}
\caption{The mass-action system $(\hat{G}, \hat{\bk})$ excludes the enzyme reaction with all reaction rate constants are equal to 1.}
\label{fig2:input_ouput_mas}
\end{figure}
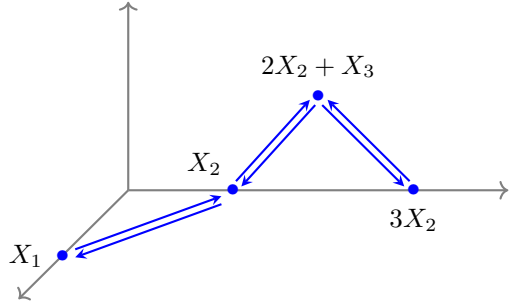
\qed
\end{example}

Let us end this section with an interesting fact contained in Remark \ref{rmk:homeo_cb_nd} below.
Recall from Theorem \ref{thm:irreducible} and Definition \ref{def:homeo_type} that homeostasis can exhibit various homeostasis types depending on the determinant of specific homeostasis blocks $B_{\eta}$ in \eqref{eq:irreducible} being zero.

In particular, when the size of the homeostasis block $B_{\eta}$ is $1 \times 1$, the block is termed \emph{Haldane} if $B_{\eta} = f_{i, j}$ with $i \neq j$. Otherwise, if $B_{\eta} = f_{i, i}$, it is termed \emph{null degradation} \cite{GS17}.

\begin{proposition}[\cite{Feinberg1987}]
\label{prop:cycle_decompose}

Let $(G, \bk)$ be a complex-balanced system with a positive equilibrium $\bx^* \in \mathbb{R}_{>0}^n$.
Then it admits a cyclic decomposition, that is,
\[
(G, \bk) = (G_1, \bk_1) \ \bigoplus \ \cdots \ \bigoplus \ (G_p, \bk_p),
\]
where for every $1 \leq i \leq p$, $G_i$ is a cycle and $\bx^*$ is a complex-balanced equilibrium of $(G_i, \bk_i)$.
\qed
\end{proposition}

\begin{lemma} 
\label{lem:cb_nd}

Let $(G, \bk)$ be a mass-action system with the stoichiometric subspace $\mS = \mathbb{R}^n$ given by
\begin{equation} \label{eq:cb_nd}
\frac{d \bx}{d t} 
= \bf (\bx) = \sum_{\by_i \rightarrow \by_j \in E}k_{\by_i \rightarrow \by_j} \bx^{\by_i}(\by_j - \by_i).
\end{equation}
Suppose $(G, \bk)$ has a complex-balanced equilibrium $\bx^* \in \mathbb{R}_{>0}^n$, then for any $1 \leq i \leq n$,
\[
\frac{\partial f_i}{\partial x_i} \neq 0 
\ \text{ at } \
\bx = \bx^*,
\]
where $\bf (\bx) = (f_1 (\bx), \ldots, f_n (\bx))$ is defined in \eqref{eq:cb_nd}.
\end{lemma}

\begin{proof}

We prove this lemma by contradiction.
Suppose that $(G, \bk)$ has a complex-balanced equilibrium $\bx^* \in \mathbb{R}_{>0}^n$, there exists $1 \leq i \leq n$ such that $\frac{\partial f_i}{\partial x_i} = 0$ at $\bx = \bx^*$.
Without loss of generality, we assume that $i = 1$ and thus 
\[
f_1 = \frac{\partial f_1}{\partial x_1} = 0
\ \text{ at } \
\bx = \bx^*.
\]

Since $(G, \bk)$ is a complex-balanced system, from Proposition \ref{prop:cycle_decompose} it can be decomposed into cycles $C_1, \ldots, C_p$ and $\bx^*$ is a complex-balanced equilibrium for every cycle.
Without loss of generality, we assume one of the cycles involved with $X_1$ as follows:
\be \notag
C_1: \by_1 \xrightarrow[]{k_{1}} \by_2 \xrightarrow[]{k_{2}} \by_3 \rightarrow \cdots \rightarrow \by_q \xrightarrow[]{k_{q}} \by_1.
\ee 
where $\by_i \in V$ and at least one of $\{\by_i\}$ has $X_1$ component.
Moreover, we have
\be \label{est1: cb in cycle}
k_{1} (\bx^*)^{\by_1} = k_{2} (\bx^*)^{\by_2} = \cdots = k_{q} (\bx^*)^{\by_q}.
\ee
Let $\Delta_i$ denote the difference of $X_1$ component on $\by_i$ and $\by_{i+1}$, that is,
\begin{equation} \label{eq:delta_i}
\Delta_i = \by_{i+1, 1} - \by_{i, 1} 
\ \text{ for } \
1 \leq i \leq q - 1
\ \text{ and } \
\Delta_q = \by_{1, 1} - \by_{q, 1}.
\end{equation}
Following the mass-action kinetics, we obtain the equation of the rate of change on $x_1$ in this cycle, given by
\[
\tilde{f}_1 = \sum\limits^{q}_{i=1} k_{i} \bx^{\by_i} \Delta_i.
\]
 From \eqref{eq:delta_i} applying the cycle property we have
\[
\sum\limits^{n}_{i=1} \Delta_i = 0.
\]
This, together with \eqref{est1: cb in cycle} implies 
\[
\tilde{f}_1 = \sum\limits^{q}_{i=1} k_{i} (\bx^*)^{\by_i} \Delta_i = 0
\ \text{ at } \
\bx = \bx^*.
\]
Now, we take the partial derivative on $\tilde{f}_1$ getting  
\be \notag
\frac{\partial \tilde{f}_1}{\partial x_1}
= \sum\limits^{q}_{i=1} 
k_{i} \bx^{\by_i} \Delta_i \times \frac{\by_{i,1}}{x_1}.
\ee
By substituting $\bx = \bx_{*}$, together with \eqref{est1: cb in cycle}, we obtain 
\be
\frac{\partial \tilde{f}_1}{\partial x_1} (\bx^*)
= \sum\limits^{q}_{i=1} k_{i} (\bx^{*})^{\by_i} \Delta_i \times \frac{\by_{i,1}}{x^*_1}
= \frac{k_{1} (\bx^{*})^{\by_1}}{x^*_1} \sum\limits^{q}_{i=1} 
 \Delta_i \by_{i,1}.
\ee
On the other hand, from \eqref{eq:delta_i} we have 
\be \notag
\begin{split}
\sum\limits^{q}_{i=1} \Delta_i \by_{i,1}
& = \sum\limits^{q-1}_{i=1} (\by_{i+1, 1} - \by_{i, 1}) \by_{i,1} + (\by_{1, 1} - \by_{q, 1}) \by_{q,1}
\\& = - \frac{1}{2} \sum\limits^{q-1}_{i=1} (\by_{i+1, 1} - \by_{i, 1})^2 - \frac{1}{2} (\by_{1, 1} - \by_{q, 1})^2 \leq 0,
\end{split}
\ee
where the equality holds when $\by_{1,1} = \ldots = \by_{q,1}$. 
Thus, $\frac{\partial \tilde{f}_1}{\partial x_1} (\bx^*) = 0$ holds when all vertices in the cycle have the same amount of $X_1$, which indicates $\tilde{f}_1 \equiv 0$. 

Analogously, we can deduce a similar result for the rest of the decomposed cycles $C_2, \ldots, C_p$. For any $1 \leq 2 \leq p$, $\frac{\partial \tilde{f}_i}{\partial x_i} (\bx^*) = 0$ holds when all vertices in the cycle $C_i$ have the same amount of $X_1$, indicating $\tilde{f}_i \equiv 0$.
Therefore, $\frac{\partial f_1}{\partial x_1} (\bx^*) = 0$ holds when all vertices in the cycle have the same amount of $X_1$ and thus 
\[
f_1 = \sum \tilde{f}_1 + \ldots \tilde{f}_p \equiv 0.
\]
However, this contradicts our assumption that the stoichiometric subspace $\mS = \mathbb{R}^n$
\end{proof}

The following remark follows from Lemma \ref{lem:cb_nd} and Deficiency Zero Theorem (Theorem \ref{thm:deficiency_zero}).

\begin{remark}
\label{rmk:homeo_cb_nd}

\begin{enumerate}
\item[(a)] If the input-output system $(G, \bk, \II)$ exhibits infinitesimal homeostasis at $\II_0$ and $(G, \bk, \II_0)$ is a complex-balanced system, then the infinitesimal homeostasis at $\II_0$ is not of the null degradation type.

\item[(b)] Suppose the input-output system $(G, \bk, \II)$ is weakly reversible and has zero deficiency. If the system exhibits infinitesimal homeostasis at $\II_0$, then it is not of the null degradation type.
\end{enumerate}
\end{remark}


\section{Main Results}
\label{sec:main_result}

Now we present the main results of this paper, providing an explicit method to compute \emph{infinitesimal homeostasis} and \emph{infinitesimal concentration robustness} on E-graphs with $\mS \subsetneq \mathbb{R}^n$.
The following settings for mass-action systems will be used throughout this section.

\medskip

Consider the mass-action system $\Gk$ as follows:
\be \label{notation:system}
\frac{d \bx}{d t} = \bf (\bx) = \sum_{\by\rightarrow \by' \in E}k_{\by \rightarrow \by'} \bx^{\by}(\by' - \by),
\ee
where $\bf = (f_1, \ldots, f_n)$.
Assume that its stoichiometric subspace $\mS \subsetneq \mathbb{R}^n$ satisfies 
\be \label{notation:dimS=n-d}
\dim (\mS) = n - d < n
\ \text{ and } \ 
\dim (\mS^{\perp}) = d > 0,
\ee
Without loss of generality, assume that $\{ f_{d + 1}, \ldots, f_{n} \}$ are linearly independent and 
\be \label{notation:n-d_linear_ind}
f_i \in \spn \{ f_{d + 1}, \ldots, f_{n} \}
\ \text{ for } \
1 \leq i \leq d.
\ee
Further, denote an orthonormal basis of $\mS^{\perp}$ by 
\be \label{notation:basis_S_perp}
U = \{ \bu_1, \ldots, \bu_d \} \subset \mathbb{R}^n.
\ee
where $\bu_i = (\bu_{i, 1}, \ldots, \bu_{i, n})^{\intercal}$ for $1 \leq i \leq d$. Note that from \eqref{eq:linear_conservation} it follows that $\Gk$ has linear conservation laws, such that
\be \label{notation:conservation_law}
\bu_i \cdot \frac{d \bx}{d t} = 0
\ \text{ for } \
1 \leq i \leq d.
\ee
\qed

\subsection{Reduced Jacobian Matrix}
\label{ss:reduce_jacobian}

In Section \ref{sec:basic}, we focused on mass-action systems with $\mS = \mathbb{R}^n$, where infinitesimal homeostasis is well-defined around a hyperbolic equilibrium.

However, the Jacobian matrix is degenerate when $\mS \subsetneq \mathbb{R}^n$, indicating that the system cannot guarantee a hyperbolic equilibrium.
Therefore, we introduce the \emph{reduced Jacobian matrix}, which plays a crucial role in establishing our results on infinitesimal homeostasis and infinitesimal concentration robustness.

\begin{definition}
\label{def:reduced_jacobian}

Let $(G, \bk)$ be a mass-action system as defined in \eqref{notation:system}. Suppose it satisfies \eqref{notation:dimS=n-d}-\eqref{notation:basis_S_perp} and the conservation laws \eqref{notation:conservation_law}. Define the \defi{reduced Jacobian matrix} $\tilde{J}$ associated with the system $(G, \bk)$ as
\be \label{eq:reduce_jacobian}
\tilde{J} = 
\begin{pmatrix} 
\bu_{1, 1} & \cdots & \bu_{1, n} \\ 
\vdots & \ddots & \vdots \\
\bu_{d, 1} & \cdots & \bu_{d, n} \\[5pt]
\frac{\p f_{d+1}}{\p x_1} & \cdots & \frac{\p f_{d+1}}{\p x_n} \\ 
\vdots & \ddots & \vdots \\
\frac{\p f_{n}}{\p x_1} & \cdots & \frac{\p f_{n}}{\p x_n}
\end{pmatrix}.
\ee
Further, let $\tilde{H}_i$ denote the \defi{$i^{th}$ reduced homeostasis matrix} obtained by removing the $i^{th}$ row and the $n^{th}$ column from $\tilde{J}$.
\end{definition}

Suppose the mass-action system $\Gk$ in \eqref{notation:system} satisfies \eqref{notation:dimS=n-d}-\eqref{notation:basis_S_perp} and the conservation laws \eqref{notation:conservation_law}.
Assume $\Gk$ has an equilibrium $\bx = \bx^*$, then \eqref{notation:dimS=n-d} and \eqref{notation:n-d_linear_ind} implies
\[
\ker (J) \subseteq \mS^{\perp}.
\]
Hence, from \eqref{notation:n-d_linear_ind} we assume that for $1 \leq i \leq d$,
\be \label{eq2:reduce_system}
f_i = \sum\limits^{n}_{j=d+1} a_{i,j} f_j
\ \text{ with } \
a_{i, j} \in \rr.
\ee
From \eqref{notation:system}, this indicates that
\be \notag
\frac{ d x_i}{ dt } 
= \sum\limits^{n}_{j=d+1} a_{i,j} \frac{ d x_j}{ dt }
\ \text{ for } \
1 \leq i \leq d.
\ee
Together with the conservation laws \eqref{notation:conservation_law}, the above implies that there exist constants $c_1, \ldots, c_{n-d}$, such that
\be \label{eq:conservation_law}
x_i = \sum\limits^{n}_{j=d+1} a_{i,j} x_j + c_i
\ \text{ for } \
1 \leq i \leq d.
\ee
In view of \eqref{eq:conservation_law}, we may write 
\[
\tilde{f}_{j} (x_{d+1}, \ldots, x_n)
= f_{j} (x_1, \ldots, x_n)
\ \text{ for } \
1 \leq i \leq d.
\]
Define the \defi{reduced mass-action system} associated with $(G, \bk)$ as follows:
\be \label{eq:reduce_system}
\begin{split}
\frac{d x_{d+1}}{ dt } 
& = \tilde{f}_{d+1} (x_{d+1}, \ldots, x_n),
\\ & \ \vdots
\\ \frac{d x_{n}}{ dt } \ 
& = \tilde{f}_{n} (x_{d+1}, \ldots, x_n).
\end{split}
\ee
Note that $\bx = \bx^*$ is also an equilibrium of the reduced mass-action system \eqref{eq:reduce_system}. 

Furthermore, the equilibrium $\bx^*$ is called a \defi{reduced hyperbolic equilibrium} of $\Gk$ if the Jacobian matrix of \eqref{eq:reduce_system} at $\bx = \bx^*$ is non-degenerate, that is,
\be \label{eq3:reduce_system}
\det
\begin{pmatrix} 
\frac{\p \tilde{f}_{d + 1}}{\p x_{d + 1}} & \cdots & \frac{\p \tilde{f}_{d + 1}}{\p x_n} \\  
\vdots & \ddots & \vdots \\
\frac{\p \tilde{f}_{n}}{\p x_{d + 1}} & \cdots & \frac{\p \tilde{f}_{n}}{\p x_n}
\end{pmatrix} 
\neq 0 
\ \text{ at } \
\bx = \bx^*.
\ee

\begin{lemma}
\label{lem:linear_stable}

Consider the mass-action system $\Gk$ in \eqref{notation:system}.
Suppose it satisfies \eqref{notation:dimS=n-d}-\eqref{notation:basis_S_perp} and the conservation laws \eqref{notation:conservation_law}. Assume the system has a reduced hyperbolic equilibrium
$\bx = \bx^*$, then
\be \label{eq:ker_J_S}
\ker (J) \cap \mS = \{ \mathbf{0} \}
\ \text{ at } \
\bx = \bx^*.
\ee
where $J$ denotes the Jacobian matrix of $\Gk$.
\end{lemma}

\begin{proof}

Since the mass-action system $\Gk$ satisfies \eqref{notation:n-d_linear_ind}, then
\be \notag
f_i \in \spn \{ f_{d+1}, \ldots, f_{n} \}
\ \text{ for } \
1 \leq i \leq d.
\ee
Without loss of generality, we assume \eqref{eq2:reduce_system} and \eqref{eq:conservation_law}, then we obtain the reduced mass-action system \eqref{eq:reduce_system} satisfying
\be \label{eq4:reduce_system}
\frac{\p \tilde{f}_{j}}{\p x_{i}}
= \frac{\p f_{j}}{\p x_{1}} a_{1, i} + \cdots + \frac{\p f_{j}}{\p x_{d}} a_{d, i} + \frac{\p f_{j}}{\p x_{i}}.
\ee
Furthermore, since $\Gk$ has a reduced hyperbolic equilibrium $\bx = \bx^*$, then the Jacobian matrix of \eqref{eq:reduce_system} at $\bx = \bx^*$ is non-degenerate, that is,
\be \label{eq6:reduce_system}
\ker
\begin{pmatrix} 
\frac{\p \tilde{f}_{d + 1}}{\p x_{d + 1}} & \cdots & \frac{\p \tilde{f}_{d + 1}}{\p x_n} \\  
\vdots & \ddots & \vdots \\
\frac{\p \tilde{f}_{n}}{\p x_{d + 1}} & \cdots & \frac{\p \tilde{f}_{n}}{\p x_n}
\end{pmatrix} 
= \{ \mathbf{0} \}
\ \text{ at } \
\bx = \bx^*.
\ee

We now prove the lemma by contradiction. Suppose \eqref{eq:ker_J_S} fails, then there exists a non-zero vector $\bv = (v_1, \ldots, v_n)^{\intercal}$, such that
\be \notag
\bv \in \ker (J) \cap \mS
\ \text{ at } \
\bx = \bx^*.
\ee
Then we deduce that
\be \label{eq5:reduce_system}
\begin{pmatrix} 
\frac{\p f_{1}}{\p x_{1}} & \cdots & \frac{\p f_{1}}{\p x_n} \\[5pt]
\vdots & \ddots & \vdots \\[5pt]
\frac{\p f_{n}}{\p x_{1}} & \cdots & \frac{\p f_{n}}{\p x_n} 
\end{pmatrix} 
\cdot \begin{pmatrix} 
v_1 \\[5pt]  
\vdots \\[5pt]
v_n
\end{pmatrix} 
= \mathbf{0}
\ \text{ at } \
\bx = \bx^*.
\ee
On the other hand, $\bv \in \mS$, together with \eqref{eq2:reduce_system}, implies 
\be \label{eq6:conservation_law}
v_i = \sum\limits^{n}_{j=d+1} a_{i,j} v_j
\ \text{ for } \
1 \leq i \leq d.
\ee
Combining \eqref{eq5:reduce_system} and \eqref{eq6:conservation_law}, we obtain for any $d+1 \leq j \leq n$,
\be \label{eq7:conservation_law}
\sum\limits^{n}_{i=d+1}
\Big(
\frac{\p f_{j}}{\p x_{1}} a_{1, i} + \cdots + \frac{\p f_{j}}{\p x_{n-d}} a_{n-d, i} + \frac{\p f_{j}}{\p x_{i}}
\Big) v_i
= 0.
\ee
This, together with \eqref{eq4:reduce_system}, gives
\be \notag
\begin{pmatrix} 
\frac{\p \tilde{f}_{d + 1}}{\p x_{d + 1}} & \cdots & \frac{\p \tilde{f}_{d + 1}}{\p x_n} \\  
\vdots & \ddots & \vdots \\
\frac{\p \tilde{f}_{n}}{\p x_{d + 1}} & \cdots & \frac{\p \tilde{f}_{n}}{\p x_n}
\end{pmatrix} 
\cdot \begin{pmatrix} 
v_{d+1} \\[5pt]  
\vdots \\[5pt]
v_n
\end{pmatrix} 
= \mathbf{0}
\ \text{ at } \
\bx = \bx^*.
\ee
Using \eqref{eq6:reduce_system} and \eqref{eq6:conservation_law}, we conclude $\bv = \mathbf{0}$ and this contradicts  the assumption that $\bv$ is non-zero.
\end{proof}

\begin{lemma}
\label{lem:reduce_jacobian_nonzero}

Consider the mass-action system $\Gk$ in \eqref{notation:system}.
Suppose it satisfies \eqref{notation:dimS=n-d}-\eqref{notation:basis_S_perp} and the conservation laws \eqref{notation:conservation_law}. Assume the system has a reduced hyperbolic equilibrium
$\bx = \bx^*$, then
\be \label{eq:reduce_jacobian_nonzero}
\det (\tilde{J}) \neq 0
\ \text{ at } \
\bx = \bx^*,
\ee
where $\tilde{J}$ is the reduced Jacobian matrix defined in \eqref{eq:reduce_jacobian}.
\end{lemma}

\begin{proof}

Since $\bx^*$ is a reduced hyperbolic equilibrium,
Lemma \ref{lem:linear_stable} shows the Jacobian matrix of $\Gk$ satisfies 
\[
\ker (J) \cap \mS = \{ \mathbf{0} \}
\ \text{ at } \
\bx = \bx^*.
\]
Together with the assumption \eqref{notation:n-d_linear_ind}, we get
\be \label{eq:ker_reduce_J_S}
\ker
\begin{pmatrix} 
\frac{\p f_{d+1}}{\p x_1} & \cdots & \frac{\p f_{d+1}}{\p x_n} \\[5pt]
\vdots & \ddots & \vdots \\[5pt]
\frac{\p f_{n}}{\p x_1} & \cdots & \frac{\p f_{n}}{\p x_n}
\end{pmatrix} 
\cap \mS = \{ \mathbf{0} \}
\ \text{ at } \
\bx = \bx^*.
\ee

We again prove this lemma by contradiction. Suppose \eqref{eq:reduce_jacobian_nonzero} fails, then the reduced Jacobian matrix $\tilde{J}$ has a non-trivial kernel containing $\vv v \neq \mathbf{0}$ at $\bx = \bx^*$. From \eqref{eq:ker_reduce_J_S}, we get $\bv \notin \mS$ and thus it can be expressed as
\be \notag
\bv = \bv_1 + \bv_2
\ \text{ with } \
\bv_1 \in \mS
\ \text{ and } \
\mathbf{0} \neq \bv_2 \in \mS^{\perp}.
\ee
Together with the assumption \eqref{notation:basis_S_perp} that $\{\bu_1, \ldots, \bu_d \}$ is an orthonormal basis of $\mS^{\perp}$, we obtain
\be \notag
\begin{pmatrix} 
\vv u_{1, 1} & \cdots & \vv u_{1, n} \\ 
\vdots & \ddots & \vdots \\
\vv u_{d, 1} & \cdots & \vv u_{d, n}
\end{pmatrix} \cdot \vv v =
\begin{pmatrix}
\vv u_{1, 1} & \cdots & \vv u_{1, n} \\ 
\vdots & \ddots & \vdots \\
\vv u_{d, 1} & \cdots & \vv u_{d, n}
\end{pmatrix} \cdot \vv v_2
\neq \vv 0.
\ee
This contradicts the assumption that  $\vv v \in \ker (\tilde{J})$, which concludes the proof of the lemma.
\end{proof}

\subsection{Deterministic Systems}
\label{ss:deterministic_system}

In what follows, it will be convenient to consider the mass-action system \eqref{notation:system} as a general input-output system $(G, \bk, \II)$ where now the input parameter $\II$ may take a more general form.  First, let us suppose that $\II=k_{ij}$ is a  \emph{reaction rate constant  associated with the edge $\by_i \rightarrow \by_j \in E$}, given by
\begin{equation} \label{notation:yiyj_kij}
\by_i \xrightarrow{k_{ij}} \by_j. 
\end{equation} 
Suppose that the input parameter $k_{ij}$ appears only in one of the equations in \eqref{notation:system}. Under mass-action kinetics, we further assume
\[
\by_j - \by_i \in \spn \{ \mathbf{e}_p \},
\]
(note that the enzyme reaction $\emptyset \to a X_1$ in \eqref{def:enzyme_reaction} satisfies this assumption). 
Furthermore, $x_n$ is the output parameter, and every reaction in $\Gk$ except the reaction \eqref{notation:yiyj_kij} has a fixed reaction rate constant.

\begin{definition}
\label{def:inf_homeostasis_conservtion}

Let $(G, \bk)$ be a mass-action system as defined in \eqref{notation:system}. Suppose the input-output function $x_n (k_{ij})$ is well-defined in a neighborhood of $k^*_{ij}$. \defi{Infinitesimal homeostasis} occurs at $k^*_{ij}$ if 
\be \label{eq:inf_homeostasis_conservtion}
\frac{d}{d k_{ij}} x_n (k^*_{ij}) = 0.
\ee
\end{definition}

The following is our main theorem about infinitesimal homeostasis in mass-action systems with $\mS \subsetneq \mathbb{R}^n$.

\begin{theorem}
\label{thm:inf_homeostasis_conservtion}

Consider the mass-action system $\Gk$ in \eqref{notation:system} with initial condition $\bx_0 \in \RR_{>0}^n$.
Suppose it satisfies \eqref{notation:dimS=n-d}-\eqref{notation:conservation_law} and includes the reaction in \eqref{notation:yiyj_kij}.
Assume the system has a reduced hyperbolic equilibrium
$\bx = \bx^* \in \rr^n_{>0}$ when $k_{ij} = k_{ij}^*$, and
\be \notag
\by_j - \by_i \in \spn \{ \mathbf{e}_p \}.
\ee
Then $(G, \bk)$ exhibits infinitesimal homeostasis at $k^*_{ij}$ if
\be \notag
\det ( \tilde{H_{p}} ) = 0
\ \text{ at } \
\bx = \bx^*,
\ee
where $\tilde{H}_p$ is the $p^{th}$ reduced homeostasis matrix defined in Definition \ref{def:reduced_jacobian}.
\end{theorem}

\begin{proof}

Under the mass-action system \eqref{notation:system}, we get
\be \notag
\frac{d \bx}{d t} = \bf (\bx) = \sum_{\by\rightarrow \by' \in E}k_{\by \rightarrow \by'} \bx^{\by}(\by' - \by),
\ee
where $\bf (\bx) = (f_1 (\bx), \ldots, f_n (\bx))^{\intercal}$.
Since $\by_i \to \by_j \in E$ and $\by_j - \by_i \in \spn \{ \mathbf{e}_p \}$, there exists a non-zero integer $c \in \zz \backslash \{ 0 \}$, such that
\be \notag
\by_j - \by_i = c \times \mathbf{e}_p.
\ee
This implies that $k_{ij}$ only appears in $f_p (\bx)$, that is,
\be \label{f_kij}
\frac{d}{d k_{ij} } f_{\ell} (\bx) =
\begin{cases}
c \bx^{\by_i}, & \ell = p, \\[5pt]
0, & \ell \neq p,
\end{cases}
\ee
and thus $f_{p} (\bx)$ is linearly independent to other $f_{\ell} (\bx)$ with $\ell \neq p$.
From \eqref{notation:n-d_linear_ind}, we deirve that $d+1 \leq p \leq n$.

On the other hand, the assumptions \eqref{notation:n-d_linear_ind} and \eqref{notation:conservation_law} indicate that the mass-action system $\Gk$ having an equilibrium $\bx = \bx^*$ is equivalent to solving the following system:
\be \label{eq:reduce_system_kij}
\begin{split}
& f_{d+1} (\bx) = 0, \ldots, f_n (\bx) = 0,
\\& \bu_1 \cdot \bx = \bu_1 \cdot \bx_0, \ldots, \bu_d \cdot \bx = \bu_d \cdot \bx_0.
\end{split}
\ee
We can compute that $\bx = \bx^*$ is an equilibrium of \eqref{eq:reduce_system_kij} and the corresponding Jacobian matrix is the reduced Jacobian matrix $\tilde{J}$ defined in \eqref{eq:reduce_jacobian}, that is,
\[
\tilde{J} =
\begin{pmatrix} 
\bu_{1, 1} & \cdots & \bu_{1, n} \\ 
\vdots & \ddots & \vdots \\
\bu_{d, 1} & \cdots & \bu_{d, n} \\[5pt]
\frac{\p f_{d+1}}{\p x_1} & \cdots & \frac{\p f_{d+1}}{\p x_n} \\ 
\vdots & \ddots & \vdots \\
\frac{\p f_{n}}{\p x_1} & \cdots & \frac{\p f_{n}}{\p x_n}
\end{pmatrix}.
\]
Moreover, since $\bx = \bx^*$ is a reduced hyperbolic equilibrium,
Lemma \ref{lem:reduce_jacobian_nonzero} shows that $\tilde{J}$ is non-degenerate at $\bx = \bx^*$.
Thus, the implicit function theorem implies the existence of an open interval containing $k_{ij}$, that is,
\[
U = (k^*_{ij} - \varepsilon, k^*_{ij} + \varepsilon),
\]
such that there exists a unique and smooth family of equilibria
\be \label{x_smooth}
\bx = (x_1 (k_{ij}), \ldots, x_n (k_{ij})),
\ee
solving \eqref{eq:reduce_system_kij} for every $k_{ij} \in U$.
Taking the derivative with respect to $k_{ij}$ on the system \eqref{eq:reduce_system_kij}, together with \eqref{f_kij} and \eqref{x_smooth}, we obtain that
\be \notag
\begin{pmatrix} 
\bu_{1, 1} & \cdots & \bu_{1, n} \\ 
\vdots & \ddots & \vdots \\
\bu_{d, 1} & \cdots & \bu_{d, n} \\[5pt]
\frac{\p f_{d+1}}{\p x_1} & \cdots & \frac{\p f_{d+1}}{\p x_n} \\ 
\vdots & \ddots & \vdots \\
\frac{\p f_{n}}{\p x_1} & \cdots & \frac{\p f_{n}}{\p x_n}
\end{pmatrix}
\cdot \begin{pmatrix} 
x^{\prime}_1 (k_{ij}) \\[5pt]  
\vdots \\[5pt]
x^{\prime}_n (k_{ij})
\end{pmatrix} 
= \begin{gmatrix}[p]
0 \\
\vdots \\
0 \\
- c \bx^{\by_i} \\ 
0 \\ 
\vdots \\
0
\rowops
\mult{3}{\leftarrow\text{$p^{th}$ row}}
\end{gmatrix},
\ee
where $'$ indicates differentiation with 
respect to $k_{ij}$. 
Recall that $\bx^* \in \rr^n_{>0}$ and $c \in \zz \backslash \{ 0 \}$, then 
\be \notag
c \bx^{\by_i} \neq 0.
\ee
By Cramer's rule, we obtain that
\[
x^{\prime}_n (k_{ij}) = - \frac{c \bx^{\by_i}}{\det (\tilde{J} ) } \det (\tilde{H}_p),
\]
and thus complete the proof of the theorem.
\end{proof}

The following results are consequences of Theorem \ref{thm:inf_homeostasis_conservtion} and the Deficiency Zero Theorem (Theorem \ref{thm:deficiency_zero}).

\begin{corollary}
\label{cor:inf_homeostasis_conservtion}

Consider the mass-action system $\Gk$ in \eqref{notation:system} with initial condition $\bx_0 \in \RR_{>0}^n$.
Suppose it satisfies \eqref{notation:dimS=n-d}-\eqref{notation:conservation_law} and includes the reaction in \eqref{notation:yiyj_kij}.
Assume the system has a complex-balanced equilibrium $\bx = \bx^*$ when $k_{ij} = k_{ij}^*$, and 
\be \notag
\by_j - \by_i \in \spn \{ \mathbf{e}_p \}.
\ee
Then $(G, \bk)$ exhibits infinitesimal homeostasis at $k^*_{ij}$ if the $p^{th}$ reduced homeostasis matrix satisfies
\be \notag
\det ( \tilde{H_{p}} ) = 0
\ \text{ at } \
\bx = \bx^*.
\ee
\end{corollary}

\begin{proof}

From Theorem \ref{thm:HJ}, every complex-balanced equilibrium is linearly stable in each invariant polyhedron.
Further, \cite{craciun2020structure} shows that the reduced Jacobian matrix at every complex-balanced equilibrium is non-degenerate.
Therefore, we conclude this corollary from Theorem \ref{thm:inf_homeostasis_conservtion}.
\end{proof}

\begin{corollary}

Let $G$ be a weakly reversible and deficiency zero E-graph, and let $\Gk$ be a mass-action system in \eqref{notation:system} with initial condition $\bx_0 \in \RR_{>0}^n$.
Suppose $\Gk$ satisfies \eqref{notation:dimS=n-d}-\eqref{notation:conservation_law} and includes the reaction in \eqref{notation:yiyj_kij}.
Assume $\Gk$ has a positive equilibrium $\bx = \bx^*$ when $k_{ij} = k_{ij}^*$, and 
\be \notag
\by_j - \by_i \in \spn \{ \mathbf{e}_p \}.
\ee
Then $(G, \bk)$ exhibits infinitesimal homeostasis at $k^*_{ij}$ if the $p^{th}$ reduced homeostasis matrix satisfies
\be \notag
\det ( \tilde{H_{p}} ) = 0
\ \text{ at } \
\bx = \bx^*.
\ee
\end{corollary}

\begin{proof}

From Theorem \ref{thm:deficiency_zero}, $\Gk$ is a complex-balanced system for any reaction rate vector $\bk \in \mathbb{R}^{|E|}_{>0}$, and thus $\bx = \bx^*$ is a complex-balanced equilibrium. 
This corollary directly follows from Corollary \ref{cor:inf_homeostasis_conservtion}.
\end{proof}

Up to this point, given the mass-action system $\Gk$ in \eqref{notation:system} and the associated input-output system $\GkI$, we have consistently chosen $\II$ to represent a reaction rate constant.
We will now modify this approach as follows. 
Suppose $\Gk$ satisfies \eqref{notation:dimS=n-d}-\eqref{notation:conservation_law} and the conservation laws are given by
\be \notag
\bu_i \cdot \bx (t) \equiv M_i
\ \text{ for } \
i = 1, \ldots, d.
\ee
Let the input parameter $\II$ be the conservation law constant $M_i$ (denoted by $\II = M_i$), and let $x_n$ be the output parameter. Note that all reaction rate constants in $\Gk$ are fixed.
In the following, we define the notion of  \emph{infinitesimal concentration robustness}.

\begin{definition}
\label{def:inf_robust_concentration}

Let $(G, \bk)$ be a mass-action system as defined in \eqref{notation:system}. 
Suppose it satisfies \eqref{notation:dimS=n-d}-\eqref{notation:conservation_law} and the conservation laws
\be \notag
\bu_i \cdot \bx (t) \equiv M_i
\ \text{ for } \
i = 1, \ldots, d.
\ee
For any $1 \leq i \leq d$, suppose the input-output function $x_n (M_{i})$ is well-defined in a neighborhood of $M^*_{i}$. The \defi{$i^{th}$ infinitesimal concentration robustness} occurs at $M^*_{i}$ if 
\be \label{eq:inf_robust_concentration_i}
\frac{d}{d M_{i}} x_n (M^*_{i}) = 0.
\ee
Moreover, suppose the input-output function $x_n (M_{i})$ is well-defined in a neighborhood of $M^*_{i}$ for every $1 \leq i \leq d$.
\defi{Complete infinitesimal concentration robustness} occurs at $M^* = (M^*_{1}, \ldots, M^*_{d})$, if $i^{th}$ infinitesimal concentration robustness occurs at $M^*_{i}$ for every $1 \leq i \leq d$.
\end{definition}

Now we present our main result regarding infinitesimal concentration robustness in mass-action systems with $\mS \subsetneq \mathbb{R}^n$.

\begin{theorem}
\label{thm:inf_robust_concentration}

Consider the mass-action system $\Gk$ in \eqref{notation:system}.
Suppose $\Gk$ satisfies \eqref{notation:dimS=n-d}-\eqref{notation:conservation_law} and the conservation laws are given by
\be \notag
\bu_i \cdot \bx (t) \equiv M_i
\ \text{ for } \
i = 1, \ldots, d.
\ee
Assume $\Gk$ has a reduced hyperbolic equilibrium $\bx = \bx^* \in \rr^n_{>0}$ when $M = (M^*_{1}, \ldots, M^*_{d})$. 
Then for any $1 \leq i \leq d$, $(G, \bk)$ exhibits the $i^{th}$ infinitesimal concentration robustness at $M^*_{i}$ if
\be \label{eq1:inf_robust_concentration}
\det ( \tilde{H}_{i} ) = 0
\ \text{ at } \
\bx = \bx^*,
\ee
where $\tilde{H}_p$ is the $p^{th}$ reduced homeostasis matrix defined in Definition \ref{def:reduced_jacobian}.
Consequentially, $(G, \bk)$ exhibits complete infinitesimal concentration robustness at $(M^*_{1}, \ldots, M^*_{d})$, if \eqref{eq1:inf_robust_concentration} holds for every $1 \leq i \leq d$.
\end{theorem}

\begin{proof}
 
Without loss of generality, we assume $i = 1$ and thus the input parameter is $M_1$.
Under the mass-action system \eqref{notation:system}, we get
\be \notag
\frac{ d \bx}{ dt } = \vv f (\bx) = \sum_{\by_i \to \by_j \in E} k_{ij} \bx^{\by_i} (\by_j - \by_i).
\ee
From \eqref{notation:n-d_linear_ind} and \eqref{notation:conservation_law}, $\Gk$ having an equilibrium $\bx = \bx^*$ when $M = (M^*_{1}, \ldots, M^*_{d})$ is equivalent to solving the following system:
\be \label{eq:system_irc}
\begin{split}
& f_{d+1} (\bx) = 0, \ldots, f_n (\bx) = 0,
\\& \bu_1 \cdot \bx = M^*_{1}, \ldots, \bu_d \cdot \bx = M^*_{d}.
\end{split}
\ee
Similar to the proof of Theorem \ref{thm:inf_homeostasis_conservtion}, the associated Jacobian matrix to the above system is the reduced Jacobian matrix $\tilde{J}$ defined in \eqref{eq:reduce_jacobian}.

On the other hand, since $\bx = \bx^*$ is a reduced hyperbolic equilibrium, Lemma \ref{lem:reduce_jacobian_nonzero} shows that 
\[
\det (\tilde{J} ) \neq 0
\ \text{ at } \
\bx = \bx^*.
\]
The implicit function theorem shows there exists an open interval containing $M_{1}$, that is,
\[
U = (M^{*}_{1} - \varepsilon, M^{*}_{1} + \varepsilon),
\]
such that there exists a unique and smooth family of equilibria
\be \label{x_smooth_irc}
\bx = (x_1 (M_{1}), \ldots, x_n (M_{1})),
\ee
and it solves \eqref{eq:system_irc} for every $M_{1} \in U$.
Taking the derivative with respect to $M_{1}$ on \eqref{eq:system_irc}, together with \eqref{x_smooth_irc}, we obtain that
\be \notag
\begin{pmatrix} 
\bu_{1, 1} & \cdots & \bu_{1, n} \\ 
\vdots & \ddots & \vdots \\
\bu_{d, 1} & \cdots & \bu_{d, n} \\[5pt]
\frac{\p f_{d+1}}{\p x_1} & \cdots & \frac{\p f_{d+1}}{\p x_n} \\ 
\vdots & \ddots & \vdots \\
\frac{\p f_{n}}{\p x_1} & \cdots & \frac{\p f_{n}}{\p x_n}
\end{pmatrix}
\cdot \begin{pmatrix}
x^{\prime}_1 (M_{1}) \\[5pt]  
\vdots \\[5pt]
x^{\prime}_n (M_{1})
\end{pmatrix} 
= \begin{gmatrix}[p]
1 \\ 
0 \\
\vdots \\
0
\rowops
\mult{0}{\leftarrow\text{$1^{st}$ row}}
\end{gmatrix},
\ee
where $'$ indicates differentiation with 
respect to $M_{1}$.
By Cramer's rule, we compute
\[
x^{\prime}_n ( M_{1} ) = \frac{1}{\det (\tilde{J} ) } \det (\tilde{H}_{n-d+1}),
\]
and thus conclude \eqref{eq1:inf_robust_concentration}.
\end{proof}

The following corollary is a direct consequence of Theorem \ref{thm:inf_robust_concentration}.

\begin{corollary}

Consider the mass-action system $\Gk$ which satisfies the assumptions in Theorem \ref{thm:inf_robust_concentration}. Assume the system  has a reduced hyperbolic equilibrium $\bx = \bx^* \in \rr^n_{>0}$ when $M = (M^*_{1}, \ldots, M^*_{d})$.
Consider  the following matrix:
\[
\tilde{H} =
\begin{pmatrix} 
\frac{\p f_{d+1}}{\p x_1} & \cdots & \frac{\p f_{d+1}}{\p x_{n-1}} \\ 
\vdots & \ddots & \vdots \\
\frac{\p f_{n}}{\p x_1} & \cdots & \frac{\p f_{n}}{\p x_{n-1}}
\end{pmatrix}.
\]
Then $(G, \bk)$ exhibits complete infinitesimal concentration robustness at $(M^*_{1}, \ldots, M^*_{d})$ if
\be \notag
\{ \mathbf{0} \} \subsetneq \ker ( \tilde{H} )
\ \text{ at } \
\bx = \bx^*.
\ee
\end{corollary}

\begin{proof}

Note that the matrix $\tilde{H}$ is obtained from the reduced Jacobian matrix $\tilde{J}$ of \eqref{eq:io_system} by deleting its last $d$ rows and the last column. Suppose there exists a non-zero vector $\bv \in \ker ( \tilde{H} )$ at $\bx = \bx^*$. By direct computation, for every $1 \leq i \leq d$,
\[
\begin{pmatrix}  
0 \\  
\vdots \\
0 \\
\bv
\end{pmatrix} \in \ker (\tilde{H}_{i}) 
\ \text{ at } \
\bx = \bx^*,
\]
and thus we conclude this corollary from Theorem \ref{thm:inf_robust_concentration}.
\end{proof}

In addition, analogous to the infinitesimal homeostasis argument, the following results follow directly from Theorem \ref{thm:inf_robust_concentration} and the Deficiency Zero Theorem (Theorem \ref{thm:deficiency_zero}).

\begin{remark}

Let $\Gk$ be a complex-balanced system.
Suppose $\Gk$ satisfies \eqref{notation:dimS=n-d}-\eqref{notation:conservation_law} and the conservation laws are given by
\be \notag
\bu_i \cdot \bx (t) \equiv M_i
\ \text{ for } \
i = 1, \ldots, d.
\ee
Assume the system has a complex-balanced equilibrium $\bx = \bx^*$ when $M = (M^*_{1}, \ldots, M^*_{d})$.  
Then for any $1 \leq i \leq d$, $(G, \bk)$ exhibits the $i^{th}$ infinitesimal concentration robustness at $M^*_{i}$ if
\be \notag
\det ( \tilde{H}_{i} ) = 0
\ \text{ at } \
\bx = \bx^*.
\ee
Consequentially, $(G, \bk)$ exhibits complete infinitesimal concentration robustness at $(M^*_{1}, \ldots, M^*_{d})$, if the above holds for every $1 \leq i \leq d$.
\qed
\end{remark}

\begin{remark}

Let $G$ be a weakly reversible and deficiency zero E-graph, and let $\Gk$ be a mass-action system in \eqref{notation:system}.
Suppose $\Gk$ satisfies \eqref{notation:dimS=n-d}-\eqref{notation:conservation_law} and the conservation laws are given by
\be \notag
\bu_i \cdot \bx (t) \equiv M_i
\ \text{ for } \
i = 1, \ldots, d.
\ee
Assume the system has a positive equilibrium $\bx = \bx^*$ when $M = (M^*_{1}, \ldots, M^*_{d})$.
Then for any $1 \leq i \leq d$, $(G, \bk)$ exhibits the $i^{th}$ infinitesimal concentration robustness at $M^*_{i}$ if
\be \notag
\det ( \tilde{H}_{i} ) = 0
\ \text{ at } \
\bx = \bx^*.
\ee
Consequentially, $(G, \bk)$ exhibits complete infinitesimal concentration robustness at $(M^*_{1}, \ldots, M^*_{d})$, if the above holds for every $1 \leq i \leq d$.
\qed
\end{remark}

\begin{example}
\label{ex:sir}

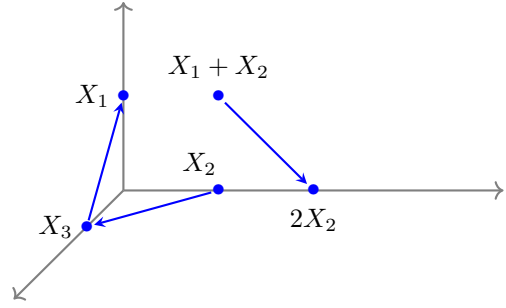
\begin{figure}[!ht]
\centering
\begin{tikzpicture}[scale=1.25]
    \draw [->, thick, gray] (0,0,0)--(4,0,0) node[right]{};
    \draw [->, thick, gray] (0,0,0)--(0,2,0) node[above]{};
    \draw [->, thick, gray] (0,0,0)--(0,0,3) node[below]{};
    \node [inner sep=0pt, blue] (1) at (1,0,0) {$\bullet$};
    \node [inner sep=0pt, blue] (11) at (0.8,0,0) {};
    \node [above=1.5pt of 11] {$X_2$};
    \node [inner sep=0pt, blue] (2) at (1,1,0) {$\bullet$};
    \node [above=1pt of 2] {$X_1 + X_2$};
    \node [inner sep=0pt, blue] (3) at (2,0,0) {$\bullet$};
    \node [below=1pt of 3] {$2 X_2$};
    \node [inner sep=0pt, blue] (4) at (0,0,1) {$\bullet$};
    \node [inner sep=0pt, blue] (41) at (0,0,1) {};
    \node [left=1pt of 41] {$X_3$};
    \node [inner sep=0pt, blue] (5) at (0,1,0) {$\bullet$};
    \node [inner sep=0pt, blue] (51) at (0,1,0) {};
    \node [left=1pt of 51] {$X_1$};
    \draw [{->}, -{stealth}, blue, thick] (1) -- (4) node [midway, below] {};
    \draw [{->}, -{stealth}, blue, thick] (2) -- (3) node [midway, above] {};
    \draw [{->}, -{stealth}, blue, thick] (4) -- (5) node [midway, below] {};
\end{tikzpicture}
\caption{SIRS Model with recovery and loss of immunity under mass-action kinetic.}
\label{fig:sir_mas}
\end{figure}

Consider the SIRS Model with recovery and loss of immunity \cite{KB23} as follows:
\be \notag
X_1 + X_2 \xrightarrow{k_1} 2 X_2, \ \
X_2 \xrightarrow{k_2} X_3, \ \
X_3 \xrightarrow{k_3} X_1,
\ee
where $X_1$, $X_2$, and $X_3$ represent the number of susceptible, infectious, and removed individuals (molecules), respectively.
In addition, $k_1$ represents the rate of spread of the disease, $k_2^{-1}$ represents the duration of the disease, and $k_3^{-1}$ represents the duration of immunity.

Construct the E-graph $G$ in Figure~\ref{fig:sir_mas}, and the associated mass-action system $\Gk$:
\begin{equation}
\begin{split} \label{eq1:sir}
& \frac{d x_1}{dt} = - k_1 x_1 x_2 + k_3 x_3,
\\& \frac{d x_2}{dt} = k_1 x_1 x_2 - k_2 x_2, 
\\& \frac{d x_3}{dt} = k_2 x_2 - k_3 x_3. 
\end{split}
\end{equation}
which admits the following conservation law:
\be \label{eq2:sir}
x_1 (t) + x_2 (t) + x_3 (t) \equiv M_1.
\ee
From \eqref{eq1:sir} and \eqref{eq2:sir}, $\Gk$ having an equilibrium $\bx = \bx^*$ when $M_1 = M^*_{1}$ is equivalent to solving the following system:
\be \label{eq3:sir}
\frac{d x_1}{dt} = - k_1 x_1 x_2 + k_3 x_3 = 0, \ \ \frac{d x_2}{dt} = k_1 x_1 x_2 - k_2 x_2 = 0, \ \
x_1 (t) + x_2 (t) + x_3 (t) = M^*_{1}.
\ee

$(a)$
Suppose the reaction rate vector $\bk = (k_1, k_2, k_3) = (1, 10, 1)$, then consider the input-output system $(G, \bk, \II)$ with $\II = M_1$, and $x_3$ as the output parameter.

By direct computation, $\bx^* = (10, 0, 0)$ is a reduced hyperbolic equilibrium of $\Gk$ when $M_1 = 10$. This kind of equilibrium is usually called a disease-free equilibrium since $x^*_2 = 0$.
From \eqref{eq3:sir}, we obtain the reduced Jacobian matrix given by 
\begin{equation} \notag
\tilde{J} = 
\begin{pmatrix} 
1 & 1 & 1 \\
- k_1 X_2  & - k_1 X_1 & k_3 \\ 
k_1 X_2  & k_1 X_1 - k_2 & 0 
\end{pmatrix}.
\end{equation}
Following Theorem \ref{thm:inf_robust_concentration}, we compute the reduced homeostasis matrix according to the conservation law \eqref{eq2:sir} as follows:
\begin{equation} \notag
H = 
\begin{pmatrix} 
- k_1 X_2  & - k_1 X_1 \\ 
k_1 X_2  & k_1 X_1 - k_2 
\end{pmatrix},
\end{equation}
and we derive that
\[
\det (H) = k_1 k_2 x^*_2 = 0
\ \text{ at } \
\bx^* = (10, 0, 0).
\]
Therefore, we conclude that $(G, \bk)$ exhibits complete infinitesimal concentration robustness at $M_1 = 10$.

\smallskip

$(b)$
Suppose the reaction rate vector $\bk = (k_1, k_2, k_3) = (1, 1, 1)$. 
Consider the same input-output system $(G, \bk, \II)$ with $\II = M_1$, and $x_3$ as the output parameter.

We can compute that $\bx^* = (1, \frac{M_1 - 1}{2}, \frac{M_1 - 1}{2})$ is the unique reduced hyperbolic equilibrium of $\Gk$ when $M_1 > 1$. 
This kind of equilibrium is usually called an endemic equilibrium since $x^*_2 > 0$.
Analogous to case $(a)$, we obtain the reduced homeostasis matrix as follows:
\begin{equation} \notag
H = 
\begin{pmatrix} 
- k_1 X_2  & - k_1 X_1 \\ 
k_1 X_2  & k_1 X_1 - k_2 
\end{pmatrix},
\end{equation}
and thus
\[
\det (H) = k_1 k_2 x^*_2 \neq 0
\ \text{ at } \
\bx^* = (1, \frac{M_1 - 1}{2}, \frac{M_1 - 1}{2}).
\]
Therefore, we conclude that $(G, \bk)$ cannot exhibit infinitesimal concentration robustness when $\bk = (1, 1, 1)$ and $M_1 > 1$.
\qed
\end{example}

\begin{remark}

The same result can be derived by looking at the form of endemic equilibrium $\bx^* = (1, \frac{M_1 - 1}{2}, \frac{M_1 - 1}{2})$, where the third component (the output variable $X_3$) is monotone with respect to the input parameter $M_1$. Therefore, it cannot exhibit infinitesimal concentration robustness at any $M_1 > 1$.    
\end{remark}

We end this subsection with an interesting result on infinitesimal homeostasis and infinitesimal concentration robustness when an E-graph consists of one reversible pair of reactions involving two species.

\begin{lemma} \label{lem: one_pair}

Consider the following mass-action system $\Gk$:
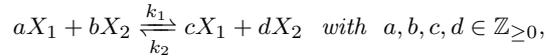
\[
a X_1 + b X_2 \xrightleftharpoons[k_2]{k_1} c X_1 + d X_2 
\ \text{ with } \
a, b, c, d \in \ZZ_{\geq 0},
\]
which admits the conservation law given by:
\be \notag
(b-d) x_1 (t) + (c-a) x_2 (t) \equiv M_1.
\ee
Assume the system has a positive equilibrium $\bx = (x^*_1, x^*_2)$ when $\bk = (k^*_1, k^*_2)$ and $M_1 = M^*_1$.
Consider the associated input-output system $\GkI$ with the output parameter $x_2$.
\begin{enumerate}
\item[(a)] If the input parameter $\II = k_i$ with $i = 1, 2$, infinitesimal homeostasis occurs at $k^*_{i}$ when $b = d$.

\item[(b)] If the input parameter $\II = M_1$, complete infinitesimal concentration robustness occurs at $M^*_{1}$ when $a = c$.
\end{enumerate}
\end{lemma}

\begin{proof}

Under the mass-action kinetic, the system $\Gk$ follows
\be \label{eq1: one_pair}
\begin{split}
& \frac{ d x_1}{ dt } = k_1 (x_1)^a (x_2)^b (c-a) + k_2 (x_1)^c (x_2)^d (a-c),
\\& \frac{ d x_2}{ dt } = k_1 (x_1)^a (x_2)^b (d-b) + k_2 (x_1)^c (x_2)^d (b-d),
\end{split}
\ee
and it admits the following conservation law:
\be \label{one_pair_conservation}
(d-b) \frac{ d x_1}{ dt } - (c-a) \frac{ d x_2}{ dt } = 0.
\ee
We exclude the case when both $a = c$ and $b = d$ hold because we always assume no self-loops in the E-graph. Note that the E-graph $G$ is weakly reversible and deficiency zero, the Deficiency Zero Theorem shows that $\bx = (x^*_1, x^*_2)$ is a complex-balanced equilibrium.
Further, \cite{craciun2020structure} shows that the reduced Jacobian matrix at every complex-balanced equilibrium is non-degenerate.

\smallskip

$(a)$
First, we consider the input parameter $\II = k_i$ with $i = 1, 2$. Without loss of generality, we assume $\II = k_1$ and $a \neq c$. Note that the system \eqref{eq1: one_pair}, together with the conservation law \eqref{one_pair_conservation}, is equivalent to the following system:
\be \label{eq2: one_pair}
\begin{split}
& \frac{ d x_1}{ dt } = f_1 (x_1, x_2) = k_1 (x_1)^a (x_2)^b (c-a) + k_2 (x_1)^c (x_2)^d (a-c),
\\& (d-b) x_1 - (c-a) x_2 = M_1,
\end{split}
\ee
whose Jacobian matrix is the reduced Jacobian matrix of \eqref{eq1: one_pair} and \eqref{one_pair_conservation}.
Since the equilibrium $\bx = (x^*_1, x^*_2)$ with $\bk = (k^*_1, k^*_2)$ and $M_1 = M^*_1$ is complex-balanced, we have shown that the corresponding Jacobian matrix of \eqref{eq2: one_pair} is non-degenerate, that is,
\[
\det
\begin{pmatrix} 
d-b & a-c \\
\frac{\partial f_1}{\partial x_1} & \frac{\partial f_1}{\partial x_2} 
\end{pmatrix}
\neq 0
\ \text{ at } \
\bx = (x^*_1, x^*_2)
\ \text{ with } \
\bk = (k^*_1, k^*_2), \ M_1 = M^*_1.
\]
Similar to the proof of Theorem \ref{thm:inf_homeostasis_conservtion}, the implicit function theorem allows us to take the derivative with respect to $k_1$ on \eqref{eq2: one_pair}, then
\be \notag
\begin{pmatrix} 
d-b & a-c \\
\frac{\partial f_1}{\partial x_1} & \frac{\partial f_1}{\partial x_2} 
\end{pmatrix}
\begin{pmatrix}
x_1^{\prime} (k_{1}) \\ 
x_2^{\prime} (k_{1})
\end{pmatrix}
= \begin{pmatrix} 
0 \\
- (x_1)^a (x_2)^b (c-a)
\end{pmatrix}
\ee
Using the Cramer's rule, we obtain $x_2^{\prime} (k^*_{1}) = 0$ if
\[
\det
\begin{pmatrix} 
d-b & 0 \\
\frac{\partial f_1}{\partial x_1} & (x_1)^a (x_2)^b (a-c)
\end{pmatrix}
= 0.
\]
Therefore, infinitesimal homeostasis occurs at $k^*_{1}$ when $b = d$.
Analogously, the same conclusion holds if the input parameter $\II = k_2$.

\smallskip

$(b)$
Second, we consider the input parameter $\II = M_1$. 
Similar to part $(a)$, the system \eqref{eq2: one_pair} is equivalent to the system \eqref{eq1: one_pair} and the conservation law \eqref{one_pair_conservation}.
Moreover, in part $(a)$ we have shown that the Jacobian matrix of \eqref{eq2: one_pair} at the equilibrium $\bx = (x^*_1, x^*_2)$ with $\bk = (k^*_1, k^*_2), \ M_1 = M^*_1$ is non-degenerate.

Similar to the proof of Theorem \ref{thm:inf_robust_concentration}, the implicit function theorem allows us to take the derivative with respect to $M_1$ on \eqref{eq2: one_pair}, then
\be \notag
\begin{pmatrix} 
d-b & a-c \\
\frac{\partial f_1}{\partial x_1} & \frac{\partial f_1}{\partial x_2} 
\end{pmatrix}
\begin{pmatrix}
x_1^{\prime} (M_{1}) \\ 
x_2^{\prime} (M_{1})
\end{pmatrix}
= \begin{pmatrix} 
1 \\ 
0
\end{pmatrix}
\ee
Using the Cramer's rule, we obtain $x_2^{\prime} (M^*_{1}) = 0$ if
\[
\det
\begin{pmatrix} 
d-b & 1 \\
\frac{\partial f_1}{\partial x_1} & 0
\end{pmatrix}
= 0.
\]
Thus, complete infinitesimal concentration robustness occurs at $M^*_{1}$ when $\frac{\partial f_1}{\partial x_1} = 0$. From the proof of Lemma \ref{lem:cb_nd}, this implies that $f_1 \equiv 0$, and we conclude that $a = c$.
\end{proof}

\subsection{Stochastic Systems}
\label{ss:stochastic_system}

In this section, we consider infinitesimal concentration robustness in stochastic reaction networks. We start by introducing \emph{stochastic mass-action systems}.
Some of our exposition follows \cite{david_kurtz_2015}. 

\begin{definition}
\label{def:sto_mass-action}

Let $(G, \bk)$ be a mass-action system as defined in \eqref{notation:system}. The \defi{associated stochastic mass-action system} generated by $(G, \bk)$ is 
\be \label{eq:sto_mass-action}
\bx (t) = \bx (0) + \sum_{\by \rightarrow \by' \in E} Y_{\by \rightarrow \by'} \big( \int^{t}_{0} \lambda_{\by \rightarrow \by'} ( \bx (s)) d s \big)
(\by' - \by),
\ee
where $\bx (t) \in \ZZ_{>0}^n$ is a vector of random variables at time $t$, and $\{ Y_{\by \rightarrow \by'} \}$ are independent unit rate Poisson processes, and
\be \notag
\lambda_{\by \rightarrow \by'} ( \bx (s) ) = k_{\by \rightarrow \by'} \prod^{n}_{i = 1} \frac{x_i !}{(x_i - y_{i}) !}.
\ee
\end{definition}

In \cite{david_kurtz_2015}, it is shown that when $\Gk$ is a complex-balanced system satisfying \eqref{notation:dimS=n-d}-\eqref{notation:conservation_law} and $\bx^*$ is a complex-balanced equilibrium, the associated stochastic mass-action system defined in \eqref{eq:sto_mass-action} admits a stationary distribution consisting of the product of Poisson distributions, that is,
\begin{equation} \label{eq:stationary distribution}
\pi (\bx) = \prod^{n}_{i = 1} \frac{(x^*_i)^{x_i}}{x_i !} e^{- x^*_i}.
\end{equation}
Furthermore, from \eqref{notation:conservation_law} and \eqref{eq:sto_mass-action} the expectation of the distribution of $\bx (t)$ satisfies
\be \notag
\bu_i \cdot E (\bx (t)) \equiv M_i
\ \text{ for } \
i = 1, \ldots, d.
\ee
Let the input parameter be the conservation law constant $M_i$, and let the output parameter be the expectation of the stationary distribution of the species $x_n$ (denoted by $E (x_n)$ or $E_n$).
This allows us to define \emph{infinitesimal concentration robustness} in stochastic mass-action systems as follows.

\begin{definition}
\label{def:irc_sto}

Let $\Gk$ be a complex-balanced system that satisfies \eqref{notation:dimS=n-d}-\eqref{notation:conservation_law}.
Consider the associated stochastic mass-action system with the distribution of $\bx (t)$ satisfying
\be \label{notation:conservation_law_sto}
\bu_i \cdot E (\bx (t)) \equiv M_i
\ \text{ for } \
i = 1, \ldots, d.
\ee
Then $\bx (t)$  admits a stationary distribution and let $E_n$ be the expectation of the stationary distribution of the species $x_n$.
For any $1 \leq i \leq d$, suppose the input-output function $E_n (M_{i})$ is well-defined in a neighborhood of $M^*_{i}$. The \defi{$i^{th}$ infinitesimal concentration robustness} occurs at $M^*_{i}$ if 
\be \label{eq:irc_i_sto}
\frac{d}{d M_{i}} E_n (M^*_{i}) = 0.
\ee
Moreover, suppose the input-output function $E_n (M_{i})$ is well-defined in a neighborhood of $M^*_{i}$ for every $1 \leq i \leq d$.
\defi{Complete infinitesimal concentration robustness} occurs at $M^* = (M^*_{1}, \ldots, M^*_{d})$, if $i^{th}$ infinitesimal concentration robustness occurs at $M^*_{i}$ for every $1 \leq i \leq d$.
\end{definition}

\begin{definition}[{\cite[Definition 2.8]{david_kurtz_2015}}]
\label{def:first_order}

An E-graph $G$ is termed a \defi{first-order E-graph} if all reactions are of the form
\[
\emptyset \to *
\ \text{ or } \
X_i \to *
\]
where $*$ represents any linear combination of the species. 
\end{definition}

\begin{lemma}[\cite{david_kurtz_2015}]
\label{lem:first_order}

Let $G$ be a first-order E-graph, and let $\Gk$ be a complex-balanced system that satisfies \eqref{notation:dimS=n-d}-\eqref{notation:conservation_law} and the conservation laws are given by
\be \notag
\bu_i \cdot \bx (t) \equiv M_i
\ \text{ for } \
i = 1, \ldots, d.
\ee
Suppose $\bx^*$ is the positive equilibrium for $\Gk$ when $M = M^*$.
Consider the associated stochastic mass-action system, then it admits a stationary distribution $\pi_{M^*} (\bx)$ that satisfies
\begin{equation} \notag
E [\pi_{M^*} (\bx)] = \bx^*.
\end{equation}
\end{lemma}

\begin{proof}

For completeness, we sketch the proof of this lemma. Since $G$ is first-order, from Definition \ref{def:first_order} there are two kinds of reactions.
First, when the reaction is $\emptyset \xrightarrow[]{k} *$, the associated intensity function is the reaction rate constant, that is,
\[
\lambda_{\emptyset \to *} = k.
\]
Second, when the reaction is $X_i \xrightarrow[]{k} *$, the associated intensity function is
\[
\lambda_{X_i \to *} = k x_i.
\]
Hence, for every reaction $\by \to \by'$ belonging to the above two cases, we obtain
\[
E [\lambda_{\by \to \by'} ( \bx (s))] = \lambda_{\by \to \by'} (E [ \bx (s)])
\ \text{ for any }
s > 0.
\]
From \eqref{eq:sto_mass-action}, this implies that $E[\bx(t)]$ is a solution to the complex-balanced system $(G, \bk)$.
Together with the conservation laws in the stochastic system below
\be \notag
\bu_i \cdot E (\bx (t)) \equiv M^*_i
\ \text{ for } \
i = 1, \ldots, d,
\ee
then we conclude this lemma.
\end{proof}

\begin{theorem}
\label{thm:irc_sto}

Let $G$ be a first-order E-graph, and let $\Gk$ be a complex-balanced system that satisfies \eqref{notation:dimS=n-d}-\eqref{notation:conservation_law} and the conservation laws are given by
\be \notag
\bu_i \cdot \bx (t) \equiv M_i
\ \text{ for } \
i = 1, \ldots, d.
\ee
Suppose $\bx^*$ is the positive equilibrium for $\Gk$ when $M = M^*$.
Consider the associated stochastic mass-action system.
Then for any $1 \leq i \leq d$, it exhibits the $i^{th}$ infinitesimal concentration robustness at $M^*_{i}$ if
\be \notag
\det ( \tilde{H}_{i} ) = 0
\ \text{ at } \
\bx = \bx^*.
\ee
Consequentially, it exhibits complete infinitesimal concentration robustness at $(M^*_{1}, \ldots, M^*_{d})$, if the above holds for every $1 \leq i \leq d$.
\end{theorem}

\begin{proof}

From Lemma \ref{lem:first_order}, for $M = M^*$ the stationary solution $\pi_{M} (\bx)$ satisfies
\begin{equation} \notag
E [\pi_{M^*} (\bx)] = \bx^*.
\end{equation}
Recall Theorem \ref{thm:HJ}, $\bx^*$ is a complex-balanced equilibrium that is linearly stable in each invariant polyhedron.
Therefore, we conclude this theorem via Theorem \ref{thm:inf_robust_concentration}.
\end{proof}

\begin{corollary}
\label{cor:irc_sto}

Let $G$ be a first-order E-graph, and let $\Gk$ be a complex-balanced system that satisfies \eqref{notation:dimS=n-d}-\eqref{notation:conservation_law} and the conservation laws are given by
\be \notag
\bu_i \cdot \bx (t) \equiv M_i
\ \text{ for } \
i = 1, \ldots, d.
\ee
Suppose $\Gk$ admits a positive equilibrium when $M = M^*$.
\begin{enumerate}
\item[(a)] For any $1 \leq i \leq d$, assume the associated stochastic mass-action system exhibits the $i^{th}$ infinitesimal concentration robustness at $M^*_{i}$.
Then it exhibits the $i^{th}$ infinitesimal concentration robustness at every $\widetilde{M}_{i}$, where $\Gk$ admits a positive equilibrium when $M = \widetilde{M}$.

\item[(b)] If the associated stochastic mass-action system exhibits complete infinitesimal concentration robustness at $M^*$, then it exhibits complete infinitesimal concentration robustness at every $\widetilde{M}$, where $\Gk$ admits a positive equilibrium when $M = \widetilde{M}$.
\end{enumerate}
\end{corollary}

\begin{proof}

From Theorem \ref{thm:irc_sto}, the stochastic system exhibiting the $i^{th}$ infinitesimal concentration robustness at $M^*_{i}$ shows that
\be \notag
\det ( \tilde{H}_{i} ) = 0
\ \text{ at } \
\bx = \bx^*.
\ee
Since $G$ is a first-order E-graph, the Jacobian matrix is fixed under any equilibrium.
Given any conservation law $M = \widetilde{M}$, $\Gk$ admits a unique complex-balanced equilibrium $\tilde{\bx}$.
Therefore, 
\be \notag
\det ( \tilde{H_i} ) = 0
\ \text{ at } \
\bx = \tilde{\bx}.
\ee
and thus we conclude it exhibits the $i^{th}$ infinitesimal concentration robustness at every $\widetilde{M}_{i}$.
We omit the proof of part $(b)$ as it is very similar to part $(a)$.
\end{proof}


\section{Discussion}
\label{sec:discussion}

In this paper, we have derived some methods that allow us to establish the properties of infinitesimal homeostasis and infinitesimal concentration robustness for reaction networks under mass action law in both deterministic and stochastic settings.
Our results provide an extension of several earlier ones, particularly those discussed in \cite{GS17, CD22}. 

Our work provides several new ways of expanding the theory of homeostasis to reaction network systems leading to new opportunities for applications. In particular, 
the important expansion of the classical theory of homeostasis is achieved by means of the analysis of reduced Jacobian matrix in systems with conservation laws (see Section \ref{ss:reduce_jacobian}). As we pointed out in  Example \ref{ex:sir} this method applies for instance the SIR-type models due to their important applications in epidemiology.  By studying the reduced Jacobian matrix introduced in Section \ref{ss:reduce_jacobian}, we can analyze SIR and similar models and compute infinitesimal concentration robustness without relying on an E-graph representation of a reaction network.

While beyond the scope of our current work, the approach proposed here holds potential for application to reaction networks with multiple input parameters. For instance, a reaction rate constant and a conservation law constant could be jointly used as input parameters. Therefore, it is of interest to determine under what conditions the chosen output variable can exhibit infinitesimal homeostasis with respect to multiple inputs. Some relevant ideas have been recently proposed in this area by \cite{GS18, MA22, MA23}. 
Besides multiple parameter inputs, there are also other intriguing generalizations of the notion of homeostasis in reaction networks. The first one is inspired by bifurcation theory. As discussed in this paper, the infinitesimal homeostasis occurs at $\II_0$ if $x_n'(\II_0) = 0$, where $\II \mapsto x_n(\II)$ is the input-output function. Furthermore, if $x_n''(\II_0) = 0$ and $x_n'''(\II_0) \neq 0$, then the output variable exhibits chair homeostasis at $\II_0$ \cite{GS18}. It is therefore worthwhile to study chair homeostasis in reaction networks both without and under conservation laws. In this context, it appears that our two theorems from Section~\ref{sec:main_result} can be also extended to study the chair homeostasis.

The second area for potential future expansion is the investigation of homeostasis in general complex-balanced stochastic systems. Although we presented some basic results on such systems in 
Section~\ref{ss:stochastic_system}, these were limited to stochastic reaction networks with first-order E-graphs. Extending this to networks with source complexes consisting of more than one species (e.g., SIR models) appears feasible. Since it is crucial to understand how infinitesimal homeostasis occurs in general networks, we hope to be able to pursue such work in the future.


\section*{Acknowledgements}

The authors' work was supported in part by the US  National Science Foundation under the projects DMS 2310816 and  DMS 1853587.





\end{document}